\documentclass[draft]{article}
\usepackage[a4paper,margin=25mm]{geometry}
\usepackage{amsmath,amssymb,tikz}

\parskip\smallskipamount

\let\set\mathbb
\def\<#1>{\langle#1\rangle}

\usepackage{amsthm,amsmath}
\usepackage{tikz}
\usepackage{xfrac}
\usepackage{cite}

\usepackage{mathtools}
\usepackage{enumerate}

\usepackage{mathtools}
\usepackage{enumerate}

\def\disp{\operatorname{disp}}

\newtheorem*{lemma*}{Lemma}

\newtheorem{theorem}{Theorem}
\newtheorem{prop}[theorem]{Proposition}
\newtheorem{thm}[theorem]{Theorem}
\newtheorem{corollary}[theorem]{Corollary}
\newtheorem{lemma}[theorem]{Lemma}
\newtheorem{remark}[theorem]{Remark}

\newtheorem{defi}[theorem]{Definition}
\newtheorem{example}[theorem]{Example}

\newcommand{\blue}{\color{blue}}
\newcommand{\black}{\color{black}}

\makeatletter
\newcommand{\myitem}[1]{%
	\item[(#1)]\protected@edef\@currentlabel{#1}%
}
\makeatother

\def\eatspace#1{#1}
\def\step#1#2{\par\kern1pt\hangindent#2em\hangafter=1\noindent\rlap{\small#1}\kern#2em\relax\eatspace}

\let\set\mathbb
\def\<#1>{\langle#1\rangle}

\def\OO{\mathcal{O}}

\def\val{\operatorname{val}}

\def\sol{\operatorname{Sol}}
\def\disc{\operatorname{Disc}}

\def\diag{\operatorname{diag}}

\def\lc{\operatorname{lc}}

\def\a{\alpha}

\begin{document}
\let\completeproof=f

\title{Reduction-based Creative Telescoping for P-recursive Sequences via Integral Bases
	\thanks{S.\ Chen was partially supported by the NSFC grants (No.\ 12271511 and No.\ 11688101), CAS Project for Young Scientists in Basic Research (Grant No.\ YSBR-034), the Fund of the Youth Innovation Promotion Association (Grant No.\ Y2022001), CAS, and the National Key Research and Development Project 2020YFA0712300. L.\ Du was supported by the Austrian FWF grant P31571-N32. M.\ Kauers was supported by the
	Austrian FWF grants P31571-N32 and I6130-N. R.\ Wang was supported by the NSFC grants (No. 12101449 and No. 12271511) , and the Natural Science Foundation of Tianjin, China (No. 22JCQNJC00440).
	}}

\author{  \bigskip
	Shaoshi Chen$^{1,2}$,\,  Lixin Du$^3$,\, Manuel Kauers$^3$,\, Rong-Hua Wang$^4$\\
	\normalsize
	$^1$KLMM,\, Academy of Mathematics and Systems Science,\\ \normalsize Chinese Academy of Sciences, Beijing, 100190, China.\\
	\normalsize $^2$School of Mathematical Sciences, University of Chinese Academy of Sciences, \\\normalsize Beijing, 100049, China\\
	\normalsize $^3$Institute for Algebra, Johannes Kepler University, Linz, A4040,  Austria\\
	\normalsize $^4$School of Mathematical Sciences, Tiangong University, Tianjin, 300387, China\\
	{\sf \normalsize schen@amss.ac.cn,  lixindumath@gmail.com, manuel.kauers@jku.at, wangronghua@tiangong.edu.cn}\\ \\
	Dedicated to the memory of Marko Petkov\v sek.
}
\date{}
\maketitle
\begin{abstract}
  We propose a way to split a given bivariate P-recursive sequence into a summable part
  and a non-summable part in such a way that the non-summable part is minimal in some
  sense. This decomposition gives rise to a new reduction-based creative telescoping
  algorithm based on the concept of integral bases.

  \medskip
  \emph{Keywords.} Abramov-Petkov{\v s}ek reduction; P-recursive sequences; Symbolic summation.
\end{abstract}

\section{Introduction}

With their book $A=B$~\cite{PWZbook1996}, Petkov\v sek, Wilf, and Zeilberger marked a milestone in the development of symbolic summation.
At the beginning of the 1990s, with the appearance of creative telescoping and Petkov\v sek's algorithm, it seemed that the
last word on summation in finite terms had been spoken. The machinery described in the book can decide whether a given a definite
hypergeometric sum can be written as a linear combination of hypergeometric terms. This is not only a remarkable theoretical
result, but it also quickly became an indispensable tool that saves people working in all areas of discrete mathematics from the
struggle of simplifying hypergeometric sums by hand. However, from today's perspective, it seems that the pioneering work documented in
$A=B$ marks rather the beginning than the end of a period of research on algorithms for summation and integration. In fact, the
development of symbolic summation has continued in several directions.

One direction of research concerns the generalization of the techniques to more sophisticated types of sums and
integrals~\cite[etc.]{Zeilberger1990,chyzak98,chyzak00,chyzak09a,raab12,koutschan13,schneider16,schneider17,schneider21}. % zb90, chyzak/salvy, chyzak, christoph, koepf, clemens, carsten, chyzak/salvy/kauers, abramov-bronstein-petkovsek-schneider
Thanks to this line of research, we can deal not only with hypergeometric sums, but also with summands and integrands
that are defined in terms of systems of linear differential and difference equations (D-finite functions), summands
and integrands that are defined in terms of differential or difference fields (elementary functions, $\Pi\Sigma$~expressions),
among other classes of functions.

A second direction of research concerns the refinement of the classical algorithms so as to obtain more precise information about
a given summation problem. For example, while Gosper's algorithm~\cite{Gosper1978,PWZbook1996} for indefinite hypergeometric summation can
recognize that a given hypergeometric term is not summable, the Abramov-Petkov\v sek reduction~\cite{AbramovPetkovsek01} goes one step
further and extracts from any given hypergeometric term a maximal summable part. This is in a way a hypergeometric summation analog to
the classical Hermite reduction for indefinite integration of rational functions, which writes any given rational function as the sum
of an integrable rational function and a proper rational function with a squarefree denominator.

A third direction of research concerns efficiency. Algorithms for creative telescoping produce a so-called telescoper and a so-called
certificate. The certificate tends to be much larger than the telescoper and is often not needed. Starting from 2010, algorithms
were developed that can construct a telescoper without also constructing the corresponding certificate. This family of algorithms
is now known as ``reduction-based telescoping''. It was first proposed for rational function integration~\cite{bostan10b}, where it relies
on Hermite reduction. It was then generalized to hyperexponential integrands~\cite{bostan13a}, and later in two different ways to D-finite
integrands: one version relies on Lagrange's identity~\cite{bostan18a,vanderHoeven21}, another on integral
bases~\cite{chen16a,chen17a,chen23}. Also for summation,
the case of rational functions was settled first~\cite{chen12b}. Then, using Abramov-Petkov\v sek reduction, it was formulated for the
hypergeometric case~\cite{chen15a}. Meanwhile, there are also reduction-based telescoping algorithms for P-recursive functions based on
Lagrange's identity~\cite{vanderHoeven18b,BrochetBruno23}.

What is now missing to complete the analogy between the summation case and the integration case is a reduction-based algorithm for
P-recursive functions using integral bases. The purpose of the present paper is to introduce such an algorithm. The proposed algorithm
is a direct analog of our recent algorithm for D-finite functions~\cite{chen23}, with the integral bases for the differential case provided
in~\cite{kauers15b,Aldossari20} replaced by the integral bases for the shift case provided in~\cite{chen20} (plus a new notion of
integrality at infinity, cf. Sect.~\ref{sec:ibinf}). Although some of the details are somewhat technical, it turns out that there is a
full analogy between the integration case and the summation case. At the present stage, this is merely a theoretical result. It remains
to be seen how implementations of reduction-based algorithms for D-finite functions or P-recursive sequences based on Lagrange's identity
relate to the corresponding algorithms based on integral bases.

\section{Integral bases}

Let $C$ be a field of characteristic zero and $\bar C$ be the algebraic closure of $C$. Let $C(x)$ be the field of rational functions in $x$. The shift operator $\sigma$ on $C(x)$ is the automorphism such that $\sigma(f(x))=f(x+1)$ for all $f\in C(x)$. Let $C(x)[S]$ be an Ore algebra, where $S$ is the shift operator with respect to $x$ and satisfies the commutation rule $Sf=\sigma(f)S$ for all $f\in C(x)$. The difference operator $S-1$ is denoted by $\Delta$, which satisfies the rule $\Delta f = \sigma(f)\Delta + \Delta(f)$ for all $f\in C(x)$. Let $L=\ell_0 + \ell_1 S +\cdots + \ell_r S^r \in C[x][S]$ with polynomial coefficients $\ell_i\in C[x]$ and $\ell_0\ell_r \neq 0$. Then $r$ is called the {\em order} of $L$. We consider the left $C(x)[S]$-module $A= C(x)[S]/\<L>$, where $\<L>=C(x)[S]L$. When there is no ambiguity, an equivalence class $f+\<L>$ in $A$ is also denoted by~$f$. Every element of $A$ can be uniquely represented by $f = f_0 + f_1 S +\cdots + f_{r-1} S^{r-1}$ with $f_i\in C(x)$. An element $f\in A$ is called {\em summable} in $A$ if there exists $g\in A$ such that $f=\Delta g$.%The elements of $A$ are called {\em P-recursive} sequences.
\subsection{Integral elements at finite places}
Let us recall the value functions and integral bases for P-recursive sequences introduced in~\cite{chen20}. Let $\bar C((q))$ be the field of formal Laurent series in a new indeterminate~$q$. For each $\alpha \in \bar C$, the operator $L=\ell_0 + \ell_1S+\cdots + \ell_rS^r$ acts on a sequence $b\colon\alpha+\set Z\to\bar C((q))$ via
\begin{equation}\label{EQ: Lf}
	(L\cdot b)(z):=\ell_0(z+q)b(z) + \cdots + \ell_r(z+q)b(z+r)
\end{equation}
for all $z\in\alpha+\set Z$. This gives an action of the algebra $C(x)[S]$ on the space $\bar C((q))^{\alpha+\set Z}$ of all sequences $b\colon \alpha+\set Z\to \bar C((q))$. Since $\ell_0\ell_r\neq 0$,  the set
$\sol_{\alpha}(L):=\{\,b\in \bar C((q))^{\alpha+\set Z}\mid L\cdot f=0\,\}$ is a $\bar C((q))$-vector space of
dimension~$r$. It is called the {\em solution space} of $L$ at~$\alpha$. Since $q\notin \bar C$, we have $\ell_0(z+q)\ell_r(z+q)\neq0$ for every $z\in\alpha+\set Z$. So the sequences
$b_1,\dots,b_r\colon \alpha+\set Z \to \bar C((q))$ determined by the recurrence $L\cdot b_i=0$ and the initial values $b_i(\alpha+j)=\delta_{i,j}$ for $i,j=1,\dots,r$ form a basis of $\sol_\alpha(L)$, where $\delta_{i,j}$ is the Kronecker symbol.

The {\em valuation} $\nu_q(a)$ of a nonzero formal Laurent series $a\in \bar C((q))$ is the smallest integer $ m\in\set Z$
such that the coefficient $[q^m]a$ of $q^m$ in $a$ is nonzero. Set $\nu_q(0)=+\infty$. The {\em value function} $\val_\alpha\colon A\to\set Z\cup\{\infty\}$ is defined by
\begin{equation}\label{EQ:val}
\val_\alpha(f) := \min_{b\in\sol_{\alpha}(L)}\Bigl(
     \nu_q((f\cdot b)(\alpha)) - \liminf_{n\to\infty}\nu_q(b(\alpha-n))\Bigr)
\end{equation}
for all $f\in A$. The sequence $f\cdot b$ in~\eqref{EQ:val} is defined as $P_f\cdot b$ via the action~\eqref{EQ: Lf}, if $f = P_f + \<L>$ for some $P_f\in C(x)[S]$. By convention, set $\infty-\infty=\infty$. An element $f\in A$ is called {\em (locally) integral} at $\alpha\in  \bar C$ if $\val_\alpha(f)\geq 0$. Let $Z$ be a subset of $\bar C$. An element $f\in A$ is called {\em (locally) integral} at $Z$ if $\val_\alpha(f)\geq 0$ for all $\alpha\in Z$, i.e., $f$ is locally integral at all $\alpha\in Z$.

For a fixed $\alpha\in \bar C$, the set $ C(x)_\alpha=\{\,p/q \mid p,q\in C[x],\, q(\alpha)\neq 0\}$ forms a subring of $C(x)$. The set of all elements $f\in A$ that are locally integral at some fixed $\alpha\in \bar C$ forms a $ C(x)_\alpha$-module, denoted by $\OO_\alpha$. A basis of this module is called a {\em local integral basis} at $\alpha$ of $A$. For $Z\subseteq \bar C$, a basis of $A$ is called a {\em local integral basis} at $Z$ if it is a {\em local integral basis} at all $\alpha\in Z$. If $Z$ is a finite subset of $\bar C$, a local integral basis at $Z$ exists and the algorithm~\cite{chen20} for computing such a basis for P-recursive sequences is similar to that for algebraic functions~\cite{vanHoeij94} and D-finite functions~\cite{kauers15b,Aldossari20}.

%belong to the same orbit $\alpha + \set Z$ for some $\alpha\in \bar C$ if and only if their difference $\alpha_1 - \alpha_2$ is an integer.
Let $\alpha\in \bar C$. For two elements $\alpha_1,\alpha_2\in \alpha +\set Z$, we say $\alpha_1 < \alpha_2$ if $\alpha_1-\alpha_2<0$. Now we recall some relevant properties of integral bases in~\cite[Section 5.2]{chen20}. The lemma gives a way of computing $\val_\alpha$ by choosing an appropriate basis of the solution space $\sol_\alpha(L)$.

\if t\completeproof
\blue
The proof of the following two Lemmas have been shown in~\cite[Section 5.2, item~(A)]{chen20}.
\begin{lemma*}\label{Lem: liminf}
	For a fixed $\alpha\in\bar C$, let $\zeta\in \alpha + \set Z$ be such that $\ell_0\ell_r$	has no roots in $\{\eta \in \alpha+\set Z\mid \eta < \zeta\}$. For any solution $b\in \sol_\alpha(L)$, if $\min_{i=1}^r\nu_q(b(\zeta+i-1))=0$, then $\min_{i=1}^r\nu_q(b(\zeta-n+i-1))=0$ for all $n\geq 0$. Thus $\liminf_{n\to\infty}\nu_q(b(\zeta-n))=0$.
\end{lemma*}
\begin{proof}
	We only need to show $\min_{i=1}^r\nu_q(b(\zeta-n+i-1)) = \min_{i=1}^r\nu_q(\zeta-(n-1)+i-1)$ for all $n> 0$. Then the conclusion is true by induction on $n$. Since $b\in \sol_\alpha(L)$, it satisfies the recurrence
	\[\ell_0(z+q)b(z) + \ell_1(z+q)b(z+1)+\cdots + \ell_r(z+q)b(z+r)=0\]
	for all $z\in\alpha+\set Z$.
	Substituting $z=\zeta-n$ yields that
	\[b(\zeta-n) = - \frac{1}{\ell_0(\zeta-n+q)} \left(\ell_1(\zeta-n+q)b(\zeta-n+1)+\cdots + \ell_r(\zeta-n+q)b(\zeta-n+r)\right).\]
	For $n>0$, since $\zeta - n$ is not a root of $\ell_0(x)$, it follows that $\nu_q(\ell_0(\zeta-n+q)) =0$. Then we have
	\[\nu_q(b(\zeta-n))\geq \min \{\nu_q(b(\zeta-n+1)),\ldots, \nu_q(b(\zeta-n+r))\}.\]
	Similarly, since
	\[b(\zeta-n+r) = -\frac{1}{\ell_r(\zeta-n+q)}\left(\ell_0(\zeta-n+q)b(\zeta-n) + \cdots + \ell_{r-1}(b(\zeta - n +r - 1))\right)\]
	and $\zeta-n$ is not a root of $\ell_r(x)$, we have
	\[\nu_q(b(\zeta-n+r))\geq\min \{\nu_q(b(\zeta-n)),\ldots, \nu_q(b(\zeta-n+r-1))\}.\]
	So $\min \{\nu_q(b(\zeta-n)),\ldots, \nu_q(b(\zeta-n+r-1))\} = \min \{\nu_q(b(\zeta-n+1)),\ldots, \nu_q(b(\zeta-n+r))\}$. This completes the proof.
\end{proof}
\black
\fi
\begin{lemma}\label{Lem: val computation}
	For a fixed $\alpha\in\bar C$, let $\zeta\in \alpha + \set Z$ be such that $\ell_0\ell_r$ 	has no roots in $\{\eta \in \alpha+\set Z\mid \eta < \zeta\}$. Let $b_1,\ldots, b_r$ be a basis of $\sol_\alpha(L)$ defined by the initial values $b_j(\zeta+i-1)=\delta_{i,j}$ for $i,j=1,\ldots, r$. Then for all $f\in A$ and $\eta\in \alpha +\set Z$,
	\[\val_\eta(f) = \min_{j=1}^r\nu_q((f\cdot b_j)(\eta)).\]
\end{lemma}
\begin{proof}
	See the argument in~\cite[Section 5.2, item (A)]{chen20}.
\end{proof}
\if t\completeproof
\blue
\begin{proof}
	For all $j=1,\ldots, r$, since $\min_{i=1}^{r}\nu_q (b_j(\zeta+i-1))=0$ by construction and $\ell_0\ell_r$ has no roots less than $\zeta$ in $\alpha + \set Z$, we have $\liminf_{n\to+\infty}\nu_q(b_j(\zeta-n))=0$ by Lemma~\ref{Lem: liminf}.
	Let $b = c_{1}b_{1} + \dots + c_{r}b_{r}$ for coefficients
	$c_{1},\dots,c_{r} \in \bar C((q))$.
	Let $v :=\min_{j=1}^{r} \nu_q (c_{j})$.
	Assume that $v=0$. Then for all $f\in A$ and $\eta \in \alpha+ \set Z$,
	$$\nu_q ((f \cdot b)(\eta)) \geq \min_{j=1}^{r} \nu_q ((f \cdot b_{j})(\eta)).$$
	Furthermore, by construction of the $b_{j}$'s,
	$b(\zeta + i - 1) = c_{i}$ for all $i \in \{1,\dots,r\}$, so $\min_{i=1}^{r} \nu_q (b(\zeta+i-1)) = 0$. Again by Lemma~\ref{Lem: liminf}, we have $\liminf_{n\to+\infty}\nu_q(b(\zeta-n))=0$.
	It follows that
	$$\quad\nu_q((f\cdot b)(\eta)) - \liminf\limits_{n\to +\infty} \nu_q (b(\eta-n)) {}\geq \min\limits_{j=1}^{r} \nu_q((f \cdot b_{j})(\eta)).$$
	
	Assume now that $v \neq 0$.
	In that case, consider $q^{-v}b = q^{-v}c_{1}b_{1} + \dots + q^{-v}c_{r}b_{r}$
	with $\min_{j=1}^{r} \nu_q (q^{-v}c_{j}) = 0$.
	From the above,
	\[\nu_q((f\cdot q^{-v}b)(\eta)) - \liminf_{n\to +\infty} \nu_q (q^{-v}b(\eta-n))\geq \min_{j=1}^{r} \nu_q((f \cdot b_{j})(\eta)).\]
	Since for all $\eta \in z + \set Z$ we have
	$\nu_q (q^{-v}b(\eta)) = \nu_q (b(\eta)) -v$
	and	$\nu_q ((f\cdot q^{-v}b)(\eta)) = \nu_q ((q^{-v} f \cdot b)(\eta))=\nu_q ((f\cdot b)(\eta)) -v$, it still holds that $$\nu_q((f\cdot b)(\eta)) - \liminf\limits_{n\to +\infty} \nu_q (b(\eta-n))\geq \min\limits_{j=1}^{r} \nu_q((f \cdot b_{j})(\eta)),$$
	so that indeed
	$\val_\eta(f) = \min_{j=1}^{r} \nu_q((f \cdot b_{j})(\eta))$.
\end{proof}
\black\fi

The following lemma arises from a discussion in~\cite[Section 5.2, item~(E)]{chen20}.
\begin{lemma}\label{Lem: left bound}
	Let $\zeta\in \alpha + \set Z$ with $\alpha\in \bar C$ be such that $\ell_0\ell_r$ has no roots in $Z:=\{\eta \in \alpha+\set Z\mid \eta < \zeta\}$. Then $\{1,S,\ldots, S^{r-1}\}$ is a local integral basis of $A$ at $Z\cup\{\zeta\}$.
\end{lemma}
\begin{proof}
	Let $\eta$ be an arbitrary but fixed point in $Z\cup\{\zeta\}$ and $b_1,\ldots, b_r$ be a basis of $\sol_\alpha(L)$ such that $b_j(\eta+i-1)=\delta_{i,j}$ for all $i,j=1,\ldots, r$. Since $\ell_0\ell_r$ has no roots less than $\eta$ in $\alpha+\set Z$, by Lemma~\ref{Lem: val computation} we obtain that for all $f\in A$,
	\begin{equation}\label{EQ: val_eta}
		\val_\eta(f) =\min_{j=1}^r\nu_q((f\cdot b_j)(\eta)).
	\end{equation}
	So $\val_\eta(S^{i-1}) = \min_{j=1}^r\nu_q((S^{i-1}\cdot b_j)(\eta))=\min_{j=1}^r\nu_q(b_j(\eta+i-1))=0$ and hence $S^{i-1}$ is integral at $\eta$ for all $i=1,\ldots, r$.
	
	Let $f=f_0 + f_1 S+ \cdots +f_{r-1}S^{r-1}\in A$ with $f_0, \ldots, f_{r-1}\in C(x)$. By the construction of the $b_j$'s, we have $(f\cdot b_j)(\eta) = \sum_{i=1}^rf_{i-1}(\eta+q)b_j(\eta+i-1) = f_{j-1}(\eta+q)$ for all $j=1,\ldots, r$. It follows that
	\begin{equation}\label{EQ: val_eta 2}
		\min_{j=1}^r\nu_q((f\cdot b_j)(\eta)) =\min_{j=1}^{r}\nu_q(f_{j-1}(\eta+q)).
	\end{equation}
    Assume $f$ is integral at $\eta$. Then $\val_\eta(f)\geq 0$ by definition. By \eqref{EQ: val_eta} and~\eqref{EQ: val_eta 2}, we have $\nu_q(f_{j-1}(\eta+q))\geq 0$ for all $j=1,\ldots, r$, which implies that $f_{j-1}(x)$ has no pole at $\eta$. So $f_{j-1}\in C(x)_\eta$ for all $j=1,\ldots, r$. This proves that $\{1,S,\ldots, S^{r-1}\}$ is a local integral basis at $\eta$.
\end{proof}
Local integral bases and shift operators are related as follows.
\begin{lemma}\label{Lem: integral after shift}
 % Let $L=\ell_0 + \cdots + \ell_rS^r \in C[x][S]$ with $\ell_0\ell_r\neq 0$. Then
  Let $\a\in \bar C$. Then
  \begin{enumerate}[{(i)}]
   \item\label{it: integral} $f\in A$ is locally integral at $\a$ if and only if $Sf$ is locally integral at $\a-1$;
   \item\label{it: integra basis} $\{\omega_1,\ldots,\omega_r\}$ is a local integral basis of $A$ at $\a$ if and only if $\{S\omega_1, \ldots, S\omega_r\}$ is a local integral basis of $A$ at $\a-1$.
  \end{enumerate}
\end{lemma}
\begin{proof}
% \begin{enumerate}[{(1)}]
   %\item
   \eqref{it: integral}\, Clear by $\nu_q(((Sf)\cdot u)(\a-1)) = \nu_q((f\cdot u)(\a))$ for all $f\in C(x)[S]$ and $u \in \bar C((q))^{\a+\set Z}$.

   \eqref{it: integra basis}\, ``$\Leftarrow$": Suppose $f=\sum_{i=1}^r f_i(x) \omega_i$ with $f_i(x)\in C(x)$ is integral at $\a$. By the item~\eqref{it: integral}, we get
   \[Sf = S\sum_{i=1}^rf_i(x)\omega_i = \sum_{i=1}^r f_i(x+1) S\omega_i\]
   is integral at $\a-1$. Since $\{S\omega_1, \ldots, S\omega_r\}$ is a local integral basis at $\a-1$, it follows that $f_i(x+1)\in C(x)_{\a-1}$. So $f_i(x) \in C(x)_\a$ and hence $\{\omega_1,\ldots, \omega_r\}$ is a local integral basis at~$\a$.

   ``$\Rightarrow$": Suppose $g = \sum_{i=1}^rg_i(x)S\omega_i$ with $g_i(x)\in C(x)$ is integral at $\a-1$. Since $\ell_0\ell_r\neq 0$, both $\{1,S, \ldots, S^{r-1}\}$ and $\{S, S^2, \ldots, S^r\}$ are $C(x)$-vector space bases of $A$. Then the map sending $f$ to $Sf$ is a bijection from $A$ to itself, because $S\sum_{i=1}^rf_i(x)S^{i-1} = \sum_{i=1}^rf_i(x+1)S^i$ for all $f_i(x)\in C(x)$. There exists $f\in A$ such that $g=Sf$. On the other hand, we have $g=S\sum_{i=1}^r g_i(x-1)\omega_i$. So $f=\sum_{i=1}^rg_i(x-1)\omega_i$. By the item~\eqref{it: integral} and the assumption that $g=Sf$ is integral at $\a-1$, we obtain that $f$ is integral at~$\a$. Since $\{\omega_1,\ldots,\omega_r\}$ is a local integral basis at $\a$, it follows that $g_i(x-1)\in C(x)_\a$. Thus $g_i(x) \in C(x)_{\a-1}$. This completes the proof.
   %Let $\OO_\a$ be the set of all elements in $A$ that are locally integral at $\a$. Suppose $\{\omega_1,\ldots,\omega_r\}$ is a local integral basis of $A$. Then $\OO_\a=\sum_{i=1}^nC(x)_\a \omega_i$ is the $C(x)_\a$-module generated by $\{\omega_1,\ldots,\omega_r\}$. Taking the shift operation, by~\eqref{it:integral} we get
   %\[\OO_{\a-1}=S \OO_\a= S\sum_{i=1}^rC(x)_\a \omega_i = \sum_{i=1}^r C(x)_{\a-1}S\omega_i.\] So $\{S\omega_1,\ldots,S\omega_r\}$ is a local integral basis of $A$ at $\a-1$. The proof of the other direction is similar.
%  \end{enumerate}
\end{proof}

Abramov's dispersion function was introduced in rational summation~\cite{Abramov1971,Paule1995b}. It has a feature that the dispersion of the denominator of a rational function strictly increases by one after taking the difference. Now we introduce a local version of dispersion functions and investigate their connection with local integral bases.
%	\begin{align*}\disp_\a(p) :=& \max\{k\in \set N\mid \gcd(p,\sigma^k(p)) \neq 1\}\\
%	=&\max\{k\in \set N \mid \exists\, \xi \in \bar C \text{ such that } p(\xi) = p(\xi+k) = 0\}.
%\end{align*}

\begin{defi}
	Let $\a\in\bar C$ and $p\in C[x]$. The \emph{dispersion} of $p$ in $\a+\set Z$ is defined as
	\begin{align*}
		\disp_\a(p) =\max\{k\in \set N \mid \exists\, \xi \in \a+\set Z \text{ such that } p(\xi) = p(\xi+k) = 0\}.
	\end{align*}
	By convention, $\max \emptyset = -\infty$.
	%If $p$ has no roots in $\a+\set Z$, define $\disp_\a(p) = -\infty$.
\end{defi}

Let $W=(\omega_1,\ldots,\omega_r)$ be a vector space basis of $A$ over $C(x)$ and $f=\frac{1}{u}\sum_{i=1}^ra_i\omega_i\in A$ with $a_1,\ldots, a_r, u\in C[x]$. We write $f= \frac{aW}{u}$, where $a=(a_1,\ldots, a_r)\in C[x]^r$. The vector $SW$ is defined as $(S\omega_1,\ldots,S\omega_r)$. Throughout this subsection, we assume $\gcd(a_1, \ldots, a_r, u)=1$. Then $u$ is a denominator of $f$ with respect to $W$.

\begin{prop}\label{Prop: disp after diff}
	%Let $L= \ell_0+\cdots +\ell_rS^r\in C[x][S]$ with $\ell_0,\ell_r\neq 0$.
	Let $f= \frac{aW}{u}$ and $g=\frac{bW}{v}$, where $a,b\in C[x]^r$, $u,v\in C[x]$. Let $\alpha\in \bar C$. Suppose $f=\Delta g$ and $\disp_\a(v) \geq 0$.
	\begin{enumerate}[{(i)}]
		\item\label{it: disp1} If $W$ is a local integral basis of $A$ at $\a +\set Z$, then
		$$\disp_\a(u)\geq \disp_\a(v) + 1.$$
		\item\label{it: disp2} Let $e\in C[x]$ and $M \in C[x]^{r\times r}$ be such that $SW=\frac{1}{e}MW$. If $W$ is a local integral basis of $A$ at $Z:=\{\eta\in \a +\set Z\mid \eta \leq \beta\}$ with $\beta\in \a +\set Z$ and $\beta$ is the only possible root of $e$ in $\a + \set Z$, then
		%all roots of $\ell_0\ell_n$ contained in $\a + \set Z$ belong to $Z$, then
		$$\disp_\a((x-\beta)u)\geq 1.$$
		This means the polynomial $u$ has a root in $\a + \set Z$ that is distinct from $\beta$.
	\end{enumerate}
\end{prop}
\begin{proof}
	\eqref{it: disp1} \, A direct calculation yields
	\begin{equation}\label{EQ: f=Sg-g}
		f= \Delta g = Sg - g= \frac{\sigma(b)}{\sigma(v)}SW - \frac{b}{v}W.
	\end{equation}
	
	Since $\disp_\a(v)\geq 0$, the polynomial $v$ has a root in $\a+\set Z$. Let $\gamma$ be the minimal root of $v$ in $\a +\set Z$.
	Since $x-\gamma \mid v$, we have $\sigma(x-\gamma) = x-(\gamma-1) \mid \sigma(v)$. So $\gamma -1$ is a root of~$\sigma(v)$.
	Since $W$ is a local integral basis at $\a + \set Z$, by Lemma~\ref{Lem: integral after shift} we know that $SW$ is also a local integral basis at $\a+\set Z$. So both $W$ and $SW$ are local integral bases at $\gamma-1$. Note that $\sigma(v)$ has a root at $\gamma-1$, but $v$ does not by the minimality of $\gamma$. So $Sg=\frac{\sigma(b)}{\sigma(v)}SW$ is not integral at $\gamma-1$, but $g=\frac{b}{v}W$ is integral at $\gamma-1$. Thus $f= Sg - g$ is not integral at $\gamma-1$. This implies that its denominator $u$ has a root at $\gamma-1$.
	
	Similarly, let $\beta$ be the maximal root of $v$ in $\a +\set Z$. Then $v$ has a root at $\beta$, but $\sigma(v)$ does not (otherwise $\beta +1$ would be a root of $v$). Since both $W$ and $SW$ are local integral bases at $\beta$, it follows that $g$ is not integral at $\beta$ but $Sg$ is integral at $\beta$. So $f=Sg-g$ is not integral at $\beta$ and hence $u$ has a root at $\beta$.
	
	By the definition of dispersion, we have
	\[\disp_\a(u)\geq \beta - (\gamma-1) =\beta - \gamma +1 = \disp_\a(v) + 1.\]
	%where $\disp_\a(v)=\beta - \gamma$ by the definition of dispersion.
	
	\eqref{it: disp2} \, If $v$ has a root in $\a + \set Z$ that is less than or equal to $\beta$, let $\gamma$ be the minimal root of $v$ in $\a + \set Z$. Then $\gamma \leq \beta$. Since $W$ is a local integral basis at $Z$, by Lemma~\ref{Lem: integral after shift} we know that $SW$ is a local integral basis at $Z'=\{\eta-1\mid \eta\in Z\}$. So both $W$ and $SW$ are local integral bases at $\gamma-1$. By the same argument as in~\eqref{it: disp1}, we obtain that $u$ has a root at~$\gamma-1$, which is distinct from $\beta$.
	
	If $v$ has a root in $\a + \set Z$ that is greater than $\beta$, let $\gamma$ be the maximal root of $v$ in $\a + \set Z$. Then $\gamma > \beta$ and $\gamma$ is not a root of $\sigma(v)$. Substituting $SW= \frac{1}{e}MW$ into~\eqref{EQ: f=Sg-g} yields that
	\[f= \Delta g=\left(\frac{\sigma(b)}{\sigma(v)}\frac{1}{e}M - \frac{b}{v}\right)W = \frac{a}{u}W.\]
	Note that $\gamma$ is not a root of $\sigma(v)e$ but $\gamma$ is a root of $v$. So $u$ has a root at $\gamma$, which is distinct from $\beta$.
\end{proof}

\subsection{Integral elements at infinity}\label{sec:ibinf}
In the differential case, a local integral basis at infinity is used to reduce D-finite functions to integrands with higher valuations at infinity~\cite{chen23}. To achieve the same goal in the shift case, we formulate the notion of local integral bases at infinity for P-recursive sequences, using the framework of valued vector spaces, as in~\cite{chen20}.% First we recall the generalized series solutions of a recurrence operator at infinity. Then

The operator $L=\ell_0 + \ell_1 S +\ldots + \ell_rS^r$ with $\ell_0\ell_r\neq0$ admits $r$ linearly independent solutions of the form
\begin{equation}\label{EQ:generalized_sol}
	T=x^{-s}\,\Gamma(x)^{u/v}\phi^x\exp(p(x^{1/v}))a(x^{-1/v},\log(x)),
\end{equation}
where $v\in \set N\setminus\{0\}$, $u\in \set Z$, $s\in \bar C$, $\phi\in\bar C\setminus\{0\}$, $a\in \bar C[[x]][y]$ and $p\in \bar C[x]$ with $\deg_x(p)<v$. Such objects are called {\em generalized series solutions} of $L$ at infinity, see~\cite[Section 2.4]{Kauers23}.
Let $R$ be the set of all $\bar C$-linear combinations of series in the form~\eqref{EQ:generalized_sol}. Then $R$ equipped with the natural addition and multiplication forms a ring. Let $\sol_\infty(L)$ be the set of all series solutions of $L$ in $R$. It is called the {\em solution space} of $L$ at infinity. Then $\sol_\infty(L)$ is a $\bar C$-vector space of dimension $r$.

We shall use series solutions to define a value function on $A$ at infinity. The valuation of a series at infinity will be defined in terms of the exponent $s$ in \eqref{EQ:generalized_sol}. Similar to the D-finite case~\cite{kauers15b}, we use a function $\iota$ on $\bar C$ to choose an order on these exponents.
%Throughout the paper, we assume that all series of $\sol_\infty(L)$ have $r,\phi\in C$, $p\in C[x]$ and $b\in C[[x]][y]$ (this can always be achieved by a suitable choice of $C$).

\begin{defi}\label{Def: iota}
	%Let $G$ be a totally ordered abelian group containing $\set Z$.
	Let $\iota: \bar C\to \set R$ be a function such that  for all $s,s_1,s_2\in \bar C$ and $m\in\set Q$,
	\begin{enumerate}[(i)]
		\item $\iota(s_1+s_2) \geq \iota(s_1)+\iota(s_2)$;
		\item $\iota(m+s) = \iota(m) + \iota(s)$;
		 \item $\iota(m) = m$.
	\end{enumerate}
	The {\em valuation} $\nu_\infty(t)$ of a term $t:=x^{-s}\,\Gamma(x)^{u/v}\phi^x\exp(p(x^{1/v}))\log(x)^\ell$ is defined as~$\iota(s)$. The {\em valuation} $\nu_\infty(b)$ of a nonzero series $b\in R$ at infinity is the minimum of the valuations of all the terms appearing in $b$ (with nonzero coefficients). Set $\nu_\infty(0)=\infty$. A series $b\in R$ is called {\em integral at infinity} if $\nu_\infty(b)\geq0$.
\end{defi}
%In particular, the valuation $\nu_\infty(b)$ of a nonzero rational function $b\in C(x)$ is the valuation of its Laurent series expansion at infinity. Write $b=p/q$ with $p,q\in C(x)$. Then $\nu_\infty(b)=\deg_x(q)-\deg_x(p)$ is an integer.

\begin{example}
	For $C\subseteq \overline{\set C(t)}$ with $t$ being a new indeterminate, let $[t^0]s\in \set C$ be the constant term of $s\in \bar C$ when $s$ is expanded as a Puiseux series around a fixed $\alpha \in \set C\cup\{\infty\}$. One choice of $\iota: \bar C\to \set R$ is
	$\iota(s) = \Re([t^0]s)$ for all $s\in \bar C$, where $\Re(\cdot)$ is the real part of a complex number. Unless otherwise stated, we shall always assume this choice of $\iota$ with $\alpha = \infty$ in the examples given. With this convention, $1, x^{-\sqrt{-1}}, x^{-\sqrt t}, x^{-\sqrt t -1}$ are integral at infinity, but $x^{-\sqrt{-1} +1}, x^{-\sqrt t +1}$ are not.
\end{example}

\begin{lemma}\label{Lem: valuation_series}
	The function $\nu_\infty$ on $R$ satisfies the following properties: for all $a,b\in R$,
	\begin{enumerate}[(i)]
		\item\label{it:valuation1} $\nu_\infty(a)=\infty$ if and only if $a=0$;
		\item\label{it:valuation2}  $\nu_\infty(ab)\geq \nu_\infty(a)+\nu_\infty(b)$;
		\item\label{it:valuation3}   $\nu_\infty(a+b)\geq \min \{\nu_\infty(a), \nu_\infty(b)\}$.
	\end{enumerate}
	So the set $\{b\in R \mid b \text{ is integral}\,\}$ forms a subring of $R$.
\end{lemma}
\begin{proof}
	\eqref{it:valuation1} and \eqref{it:valuation3} are clear by definition. For \eqref{it:valuation2}, note that $x^{-s_1}x^{-s_2} = x^{-(s_1 +s_2)}$ and $\iota(s_1 + s_2)\geq\iota(s_1) + \iota(s_2)$ for all $s_1,s_2\in \bar C$ by the assumption on $\iota$ in Definition~\ref{Def: iota}.
\end{proof}
The shift $\sigma: b(x)\mapsto b(x+1)$ of series is defined through the rules:
\begin{align*}
	(x+1)^{-s} &= \sum_{n=0}^{\infty}\binom{-s}{n}x^{-s-n} = x^{-s} - sx^{-s-1} + \cdots,\\
	\Gamma(x+1)^{u/v} &= x^{u/v}\Gamma(x)^{u/v},\\
	\phi^{x+1} &= \phi \phi^x,\\
	\exp(c(x+1)^{k/v}) &= \exp(cx^{k/v})\exp(c((x+1)^{k/v}-x^{k/v}))\\
	 &=\exp(cx^{k/v})\sum_{n=0}^\infty \frac{1}{n!}\left(c\sum_{i=1}^\infty \binom{k/v}{i}x^{k/v-i}\right)^n\\
	 &=\exp(cx^{k/v})\left(1+\frac{ck}{v}x^{k/v-1}+\cdots\right)\, \text{($c\in \bar C$, $0\leq k<v$)},\\
	 \log(x+1)^\ell &= \left(\log(x)+\log\left(1+\frac{1}{x}\right)\right)^\ell = \left(\log(x) - \sum_{n=1}^\infty \frac{(-1)^n}{n}x^{-n}\right)^\ell\\
	 &=\log(x)^\ell + \ell x^{-1}\log(x)^{\ell-1} + \cdots,
\end{align*}
where ``$\cdots$'' denotes some terms of higher valuations at infinity.

An operator $P = p_0 + p_1 S + \cdots +p_{r-1} S^{r-1}$ in $C(x)[S]$ acts on a series $b$ via
\[P\cdot b =p_0 b + p_1 \sigma(b) + \cdots + p_{r-1}\sigma^{r-1}(b).\]
For a solution $b\in \sol_\infty(L)\subseteq R$ and $f\in A=C(x)[S]/\<L>$, the action $f\cdot b$ is defined as $P_f\cdot b$ if $f=P_f+\<L>$ for some $P_f\in C(x)[S]$. Let $b_1, \ldots, b_r$ be a basis of $\sol_\infty(L)$ in the form of~\eqref{EQ:generalized_sol}. The {\em value function} $\val_\infty\colon A\to \set R \cup\{\infty\}$ is defined as
\begin{equation*}
	\val_\infty(f) :=\min_{i=1}^r\nu_\infty(f\cdot b_i).
\end{equation*}
Then by Lemma~\ref{Lem: valuation_series}, $\val_\infty(f)$ is the minimum valuation of all series $f\cdot b$ at infinity, where $b$ runs through series solutions in $\sol_\infty(L)$. %It can be checked that $\val_\alpha$ is indeed a value function as defined in Definition~\ref{def:value}.
An element $f\in A$ is called {\em (locally) integral} at infinity if $\val_\infty(f)\geq 0$.

\begin{prop}\label{Prop: value function}
	The function $\val_\infty$ satisfies the following properties: for all $f, g\in A$ and $a\in C(x)$,
	\begin{enumerate}[(i)]
		\item\label{it:val1} $\val_\infty(f)= \infty$ if and only if $f=0$;
		\item\label{it:val2} $\val_\infty(af)=\nu_\infty(a)+\val_\infty(f)$;
		\item\label{it:val3} $\val_\infty(f+g)\geq \min\{\val_\infty(f), \val_\infty(g)\}$;
	\end{enumerate}
    where $\nu_\infty(a)$ is the valuation of its Laurent series expansion at infinity.
\end{prop}
\begin{proof}
	\eqref{it:val1} Let $b_1, \ldots, b_r$ be a basis of $\sol_\infty(L)$ in the form of~\eqref{EQ:generalized_sol}. If $f=0$, then $f\cdot b_i$ is the zero series for all $i=1,\ldots, r$. So $\val_\infty(f) = \infty$. Conversely, assume that $\val_\infty(f) = \infty$, then by definition $\nu_\infty(f\cdot b_i) = \infty$ for all $i=1,\ldots, r$. So $f\cdot b_i =0$ for $i=1,\ldots,r$, which implies that $f$ has at least $r$ linearly independent series solutions. On the other hand, every nonzero element of $A$ can be written as $f= f_0 + f_1S+\cdots +f_kS^k$ with $f_i\in C[x]$, $f_k\neq 0$ and $k<r$. If the trailing coefficient $f_0$ of $f$ is nonzero, the dimension of the solution space of $f$ in $R$ can not exceed its order $k$, see~\cite[Section 2.4]{Kauers23}. So in this case, $f$ cannot be nonzero.  
	
	If $f_0 = 0$, we write $f = Sg$ with $g\in A$ such that the order of $g$ is $k-1$. Then $Sg\cdot b_i = 0$ and we claim that $g\cdot b_i=0$ for all $i=1,\ldots,r$. Otherwise, suppose $g\cdot b_i \neq 0$ for some $i\in\{1,\ldots,r\}$. Then $g\cdot b_i$ is a generalized series involves a term
	\[T=x^{-s}\,\Gamma(x)^{u/v}\phi^x\exp(p(x^{1/v}))\log(x)^\ell,\]
	where $v\in \set N\setminus\{0\}$, $u \in \set Z$, $\ell\in \set N$, $s\in \bar C$, $\phi\in\bar C\setminus\{0\}$ and $p\in \bar C[x]$ with $\deg_x(p)<v$. Let $T$ be the dominant term of $g\cdot b_i$,
	i.e., among all terms with minimal $s$ the one with the largest exponent~$\ell$. Then 
    \[S\cdot T= (x+1)^{-s}\,\Gamma(x+1)^{u/v}\phi^{x+1}\exp(p((x+1)^{1/v}))\log(x+1)^\ell=\phi x^{u/v}T+\cdots.\]
    Here $\phi\neq 0$ by the assumption $\ell_0\ell_r\neq 0$. So $\phi x^{u/v}T$ is the dominant term of $Sg\cdot b_i$ and hence $Sg\cdot b_i\neq 0$. This leads to a contradiction. Thus $g\cdot b_i =0$ for all $i=1,\ldots, r$. Let $j$ with $1\leq j\leq k$ be the minimal integer such that $f_j\neq 0$. We can write $f=S^j\tilde g$ with $\tilde g\in A$ such that the order of $\tilde g$ is $k-j$ and the trailing coefficient of $\tilde g$ is nonzero. We proceed similarly with $g=S^{j-1}\tilde g$, and finally obtain that $\tilde g\cdot b_i =0$ for all $i=1,\ldots,r$. So $\tilde g$ cannot be nonzero by the same argument as in the first case. Thus $f$ cannot be nonzero.
	
	\eqref{it:val2}\, Every nonzero $a\in C(x)\subseteq C[[x^{-1}]]$ can be expanded as a Laurent series in the form $a=a_mx^{-m} + a_{m-1}x^{-m-1}+\cdots$, where $a_i \in C$, $m\in \set Z$ and $a_m\neq 0$. By Definition~\ref{Def: iota}, $\nu_\infty(a) = \iota(m)= m$ is an integer. Note that for any $m\in \set Z$ and $s\in \bar C$, $x^{-m}x^{-s} = x^{-m-s}$ and $\iota(m+s) = \iota(m) +\iota(s)=m+\iota(s)$. So $\nu_{\infty}(ab)=\nu_\infty(a) + \nu_\infty(b)$ for all $a\in C(x)$ and $b\in R$. Thus~\eqref{it:val2} is true.
	
	\eqref{it:val3} Clear by the item~\eqref{it:valuation3} in Lemma~\ref{Lem: valuation_series}.
\end{proof}
For a rational function $p/q$ with $p,q\in C[x]$, we have $\nu_\infty(p/q) = \deg_x(q)-\deg_x(p)$. The set $ C(x)_\infty=\{\,p/q \mid p,q\in C[x],\, \deg_x(p)\leq \deg_x(q)\}$ forms a subring of $C(x)$. By Proposition~\ref{Prop: value function}, $(A,\val_\infty)$ is a valued vector space over the valued field $(C(x), \nu_\infty)$, see~\cite[Definition 2]{chen20}. So by Proposition 4 in~\cite{chen20}, the set of all elements $f\in A$ that are locally integral at infinity forms a $ C(x)_\infty$-module, denoted by $\OO_\infty$. A basis of this module is called a {\em local integral basis} at infinity of $A$. Similar to the algebraic and D-finite cases~\cite{vanHoeij94,Aldossari20,chen20,kauers15b}, such a basis can be computed by Algorithm 10 in~\cite{chen20}. The termination of this algorithm is guaranteed by the existence of discriminant functions. For a basis $\omega_1,\ldots, \omega_r$ of $A$, the discriminant function $\disc$ is defined as
\[\disc({\omega_1,\ldots, \omega_r })=\lfloor \nu_\infty(\det((\omega_i\cdot b_j)_{i,j=1}^r))\rfloor,\]
where $b_1,\ldots, b_r$ is a basis of $\sol_\infty(L)$ in the form of~\eqref{EQ:generalized_sol}.
It can be checked that the $\disc$ is indeed a discriminant function in~\cite[Definition 12]{chen20}.

A $C(x)$-vector space basis $\{\omega_1,\dots,\omega_r\}$ of $A=C(x)[S]/\<L>$ is
called \emph{normal at infinity} if there exist $a_1,\dots,a_r\in
C(x)$ such that $\{a_1\omega_1,\dots,a_r\omega_r\}$ is a local integral basis at infinity. Given a vector space basis, Trager's algorithm~\cite{Trager84} can be literally adapted to compute a basis that is normal at infinity and generates the same $C[x]$-module as the given basis.

For a nonconstant polynomial $g\in C[x]$, we have $\deg_x(\Delta(g)) = \deg_x(g)-1$. For a rational function $g\in C(x)$, if $\nu_\infty(g)\neq 0$, then $\nu_\infty(\Delta(g)) = \nu_\infty(g) + 1$. So the valuation at infinity of such a rational function increases by exactly one under each difference operation. In the P-recursive case, we shall prove that the valuation at infinity increases by at most one.

The product rule in the shift case says that $\Delta(uv)=\Delta(u) v + \sigma(u)\Delta(v)$ for every $u,v\in R$.
For several elements $u_1,\dots,u_n\in R$, it generalizes to
\begin{align}\label{EQ: diff n terms}
  \Delta(u_1\cdots u_n) & = \sum_{i=1}^n \sigma(u_1\cdots u_{i-1})\Delta(u_i)u_{i+1} \cdots u_n.
\end{align}

\begin{prop}\label{Prop:val_under_diff}
	Let $g\in A$. If $\val_\infty(g)\neq 0$, then $\val_\infty(\Delta g) \leq \val_\infty(g)+1$.
\end{prop}
\begin{proof}
	Let $b_i$ be a generalized series solution in $\sol_\infty(L)$ such that $\val_\infty(g)=\nu_\infty(g\cdot b_i)$.
    There $g\cdot b_i$ involves a term
	\[T=x^{-s}\,\Gamma(x)^{u/v}\phi^x\exp(p(x^{1/v}))\log(x)^\ell,\]
	where $v\in \set N\setminus\{0\}$, $u \in \set Z$, $\ell\in \set N$, $s, \phi\in \bar C\setminus\{0\}$ and $p\in \bar C[x]$ with $\deg_x(p)<v$. For this fixed series $b_i$, let $T$ be the dominant term of $g\cdot b_i$,
	i.e., among all terms with minimal $r$ the one with the largest exponent~$\ell$. Let $k = \deg_x(p)$ and $c=\lc_x(p)$ be the degree and the leading coefficient of $p$ in $x$ respectively (take $k=0$ if $p(x)=0$).
	
	Now we use Equation~\eqref{EQ: diff n terms} to calculate $\Delta \cdot T$. Note that
	\begin{align*}
	     \Delta(x^{-s})\,\Gamma(x)^{u/v}\phi^x\exp(p(x^{1/v}))\log(x)^\ell  & = -sx^{-1}T+\cdots,\\\nonumber
	     \sigma(x^{-s})\Delta(\Gamma(x)^{u/v})\phi^x\exp(p(x^{1/v}))\log(x)^\ell&=(x^{u/v}-1)T+\cdots,\\
	     \sigma(x^{-s}\,\Gamma(x)^{u/v})\Delta(\phi^x)\exp(p(x^{1/v}))\log(x)^\ell&=(\phi-1)x^{u/v}T+\cdots,\\
	     \sigma(x^{-s}\,\Gamma(x)^{u/v}\phi^x)\Delta(\exp(p(x^{1/v})))\log(x)^\ell&=\frac{\phi ck}{v}x^{u/v+k/v-1}T+\cdots,\\
	     \sigma(x^{-s}\,\Gamma(x)^{u/v}\phi^x\exp(p(x^{1/v})))\Delta(\log(x)^\ell)&=\phi\ell x^{u/v-1}\log(x)^{-1}T+\cdots.
	\end{align*}
So by Equation~\eqref{EQ: diff n terms}, we get
\[\Delta\cdot T = (-sx^{-1} -1 + \phi x^{u/v} +\frac{\phi ck}{v}x^{u/v+k/v-1}) T + \cdots. \]
	Note that $0\leq \tfrac{k}{v}<1$ because $0\leq k < v$. So $-\frac{u}{v} < -\tfrac{u}{v}-\tfrac{k}{v}+1$. Let $T_0$ be the dominant term of
	$(\Delta g)\cdot b_i$. We make the following case distinction.
	\begin{enumerate}[(1)]
		\item If $u>0$, then $-u/v < \min\left\{1, 0, -\frac{u}{v}-\frac{k}{v}+1\right\}$. We have $T_0=\phi x^{u/v}T$.
		Here $\phi\neq 0$ by the assumption $\ell_0\ell_r\neq 0$.
		\item If $u<0$, then $0<\min\left\{1, - \frac{u}{v}, -\frac{u}{v}-\frac{k}{v} + 1\right\}$. We have $T_0=-T$.
		\item If $u=0$ and $\phi \neq 1$, then $0<\min\left\{1,-\frac{k}{v}+1\right\}$. We have $T_0=(\phi-1)T$.
		\item If $u=0$, $\phi = 1$ and $0<k<v$, then $-\tfrac{k}{v} +1< 1$. We have $T_0=\frac{ck}{v} x^{k/v-1}T$.
		\item If $u=0$, $\phi = 1$ and $k=0$, then $T_0=-sx^{-1}T$. Here $s\neq 0$ by the
		assumption $\val_\infty(g)\neq 0$.
	\end{enumerate}
    The above calculation reveals that the valuation of the term $\Delta\cdot T$ in $(\Delta g)\cdot y_i$ is less than or equal to $\nu_\infty(T)+1$ by Definition~\ref{Def: iota}, which implies that $\val_\infty(\Delta g)\leq \val_\infty(g)+1$.
\end{proof}

\section{Suitable bases}

% another easy way to obtain a shift-free denominator; suitable bases

Let $W=(\omega_1,\dots,\omega_r)$ be a $C(x)$-vector space basis of $A=C(x)[S]/\<L>$. Throughout this subsection, let $e_w\in C[x]$ and $M_w=(m_{i,j})_{i,j=1}^r\in C[x]^{r\times r}$ be
such that
\[
SW= \frac{1}{e_w}M_wW,
\]
and $\gcd(e,m_{1,1},m_{1,2},\ldots,m_{r,r})=1$.
For algebraic function fields, the name ``suitable basis'' was introduced by Bronstein~\cite{bronstein98a} in the context of lazy Hermite reduction. It refers to a basis $W$ whose elements are (globally) integral and the denominator $e_w$ with respect to derivation is squarefree. In this section, we seek a basis $W$ such that $W$ is a local integral basis at almost all orbits $\alpha+\set Z\in \bar C/\set Z$ and the denominator $e_w$ (with respect to shift operation) is shift-free. Recall that a polynomial $p\in C[x]$ is said to be {\em shift-free} if $\gcd(p,\sigma^i(p))=1$ for all nonzero $i\in \set Z$. So for a shift-free polynomial $p\in C[x]$, there is no $\alpha\in \bar C$ such that $\disp_\alpha(p)>0$.

In the differential case, every (global) integral basis is suitable. So Bronstein transforms an arbitrary basis to a suitable basis, along the way of finding an integral basis. In the shift case, since there are infinitely many singularities, even ``globally" integral elements may not exist. Instead of finding a local integral basis at all $\alpha\in \bar C$, we seek for a finite subset $Z\in \bar C$ such that a local integral basis at $Z$ is ``suitable". The formal definition of suitable bases will be stated later.

We now show that the polynomial $e$ does not depend on the choice of the basis of $A$ but only on the $C[x]$-submodule it generates. Let $W$ and $U$ be two $C(x)$-vector space bases of $A$. Let $e_w, e_u\in C[x]$ and $M_w, M_u\in C[x]^{r\times r}$ be such that $SW=\frac{1}{e_w}M_wW$ and $SU=\frac{1}{e_u}M_uU$. Suppose that $W$ and $U$ generate the same submodule of $A$ over $C[x]$. Then there exists a matrix $T\in C[x]^{r\times r}$ such that $W=TU$ and $T$ is an invertible matrix over $ C[x]$. Shifting both sides of the equation, we get
\[SW=\sigma(T)SU=\left(\sigma(T)\frac{1}{e_u}M_u T^{-1}\right)W=\frac{1}{e_w}M_wW.\]
Since $\sigma(T),T^{-1}\in  C[x]^{r\times r}$,  we have $e_w$ divides $e_u$. Similarly the fact that $U=RW$ with $R=T^{-1}\in C[x]^{r\times r}$ implies that $e_u$ divides $e_w$. Thus $e_w=e_u$ when $e_w,e_u$ are monic.

For a basis $U=\{1, S, \ldots, S^{r-1}\}$ of $A$, the denominator $e_u$ may not be shift-free. However, we can transform $U$ to a local integral basis $W$ at a finite subset $Z$ of $\bar C$ such that $e_w$ is shift-free.

\begin{example}
	Let $T= \frac{x^2x!}{x+1}$, which is annihilated by $L=x^2(x+2)S-(x+1)^4\in \set C(x)[S]$. Then for $U=\{1\}$, we have $SU = \frac{(x+1)^4}{x^2(x+2)}U$. So $e_u=x^2(x+2)$ is not shift-free, because it has two roots $-2$ and $0$ in the orbit $\set Z$. Starting from $U$, we can compute a basis $W=\{(x+1)x^{-3}\}$, which is a local integral basis at $Z=\{-1,0\}$. The algorithm of computing an integral basis at $Z$ guarantees that $W$ and $U$ generate the same $\set C(x)_\a$-module for all $\a\in \set C\setminus Z$. Then we see that $SW=xW$ and hence $e_w=1$ is shift-free. In fact, if we choose an integer $\beta$ with $\beta \geq 0$ and consider $Z=\{-1,0,\ldots ,\beta\}$, then $W=\{(x+1)x^{-3}\prod_{i=1}^\beta (x-i)^{-1}\}$ is a local integral basis at $Z$ and $SW=(x-\beta)W$ with $e_w=1$ being shift-free.
\end{example}

In general, the above idea of finding a basis $W$ with $e_w$ being shift-free works as follows. Let $L=\ell_0 +\ell_1S+ \cdots + \ell_rS^r \in C[x][S]$ with $\ell_0\ell_r \neq 0$. Let $\a_1+\set Z, \ldots,\a_I +\set Z$ be distinct orbits such that for each $i$ with $1\leq i\leq I$, the product $\ell_0\ell_r$ has at least one root in $\a_i +\set Z$. Recall that for $\beta, \gamma\in \a+\set Z$ with $\a\in \bar C$, we say $\beta<\gamma$ if $\beta-\gamma<0$. For each orbit $\a_i + \set Z$, let $\a_{i,1}<\a_{i,2}<\ldots <\a_{i,J_i}$ be all roots of $\ell_0\ell_r$ in $\a_i+\set Z$. We choose an arbitrary $\beta_i\in \a_i+\set Z$ such that $\alpha_{i, J_i}\leq \beta_i$. A basis $W$ of $A=C(x)[S]/\<L>$ is called {\em suitable} at $\{\beta_1,\ldots, \beta_I\}$ if $W$ is a local integral basis of $A$ at $Z=\bigcup_{i=1}^I\{\gamma\in \a_i+\set Z \mid \a_{i,1}<\gamma\leq\beta_i\}$ and it generates the same $C(x)_\a$-module as $U=\{1,S,\ldots, S^{r-1}\}$ for all $\a\in \bar C \setminus Z$. So every suitable basis is a local integral basis at a finite subset $Z$ and the set $Z\cap (\alpha_i + \set Z)$ has an upper bound $\beta_i$ in $\alpha_i + \set Z$ for each $i$ with $1\leq i\leq I$.

Throughout the paper, we assume that if $\alpha_{i,J_i}$ and $\alpha_{i,J_{i'}}$ are conjugate over $C$ for some $1\leq i\neq i'\leq I$, then $\beta_i-\alpha_{i,J_i} = \beta_{i'} - \alpha_{i,J_{i'}}$, which implies that $\beta_i$ and $\beta_{i'}$ are also conjugate over $C$. Under this assumption, we show that all conjugates of $\beta_i$ over $C$ belong to $\{\beta_1,\ldots, \beta_I\}$. Then a suitable basis of at $\{\beta_1,\ldots, \beta_I\}$ exists and can be computed (without algebraic extension of the base field $C$) using the algorithm in~\cite{chen20} for computing an integral basis at a finite subset $Z\subseteq\bar C$. Since every irreducible polynomial in $C[x]$ is shift-free, all conjugates of $\alpha_{i,J_i}$ over $C$ belong to distinct orbits in $\{\alpha_1+\set Z, \ldots, \alpha_I+\set Z\}\subseteq\bar C/\set Z$. Note that if $p(x)\in C[x]$ is the minimal polynomial of $\alpha\in \bar C$, then $\sigma^{-k}(p(x))$ is the minimal polynomial of $\alpha+k$ for any integer~$k$. So the number of conjugates of $\beta_i$ is the same as that of $\alpha_{i,J_i}$. Thus all conjugates of $\beta_i$ over $C$ belong to $\{\beta_1,\ldots, \beta_I\}$.
%. So all roots of an irreducible polynomial in $C[x]$ belong to distinct orbits in $\bar C/\set Z$.

The following theorem says that a suitable basis has the property that $e_w$ is shift-free.

\begin{thm}\label{Thm: suitable basis}
	We use the above notations. Let $W$ be a suitable basis of $A$ at $\{\beta_1,\ldots,\beta_I\}$.
	\begin{enumerate}[{(i)}]
		\item\label{it:e_w} Let $e_w\in C[x]$ and $M_w\in C[x]^{r\times r}$ be such that $SW=\frac{1}{e_w}M_wW$. Then $e_w$ is shift-free and $\beta_1, \ldots, \beta_I$ are all possible roots of $e_w$ in $\bar C$.
		\item\label{it:f_w} Let $f_w\in C[x]$ and $N_w\in C[x]^{r\times r}$ be such that $W=\frac{1}{f_w}N_wSW$. Then $f_w$ is shift-free and $\beta_1, \ldots, \beta_I$ are all possible roots of $f_w$ in $\bar C$.
	\end{enumerate}
	In particular, $M_w$ and $N_w$ are invertible matrices over $C(x)$.
\end{thm}
\begin{proof}
	\eqref{it:e_w}~A direct calculation yields that
	\[SU=\begin{pmatrix}
		0      & 1 &    &  \\
		\vdots &   & \ddots &  \\
		0      &   &        &1\\
		-\frac{\ell_0}{\ell_r} & -\frac{\ell_1}{\ell_r}&\cdots & -\frac{\ell_{r-1}}{\ell_r}
	\end{pmatrix} U =: \frac{1}{e_u}M_uU,
	\]
	where $e_u \in C[x]$ and $M_u\in C[x]^{r\times r}$. So every root of $e_u$ must be a root of the leading coefficient $\ell_r$. Let $T\in C(x)^{r\times r}$ be such that $W=TU$. Taking the shift operation, we get
	\begin{equation}\label{EQ:SW-W}
		SW= \sigma(T)SU=\sigma(T) \frac{1}{e_u}M_uT^{-1}W=\frac{1}{e_w}M_wW.
	\end{equation}
	Since $W$ and $U$ generate the same $C(x)_\alpha$-module for all $\a \in \bar C\setminus Z$, both $T$ and $T^{-1}$ belong to $C(x)_\a^{r\times r}$. So $T$ and $T^{-1}$ have no poles at $\bar C\setminus Z$, i.e., all entries of $T$ and $T^{-1}$ have no poles at $\alpha\in \bar C\setminus Z$. Then $\sigma(T)$ has no poles at $\bar C\setminus Z'$, where $Z'=\{\gamma -1 \mid \gamma \in Z\}=\bigcup_{i=1}^I\{\gamma\in \a_i+\set Z \mid \a_{i,1}\leq \gamma \leq \beta_i-1\}$. By~\eqref{EQ:SW-W} we have
	\begin{align}\label{EQ:roots of e_w}
		\{\text{roots of $e_w$} \}&\subseteq \{\text{poles of $\sigma(T)$}\}\cup\{\text{roots of $e_u$}\}\cup\{\text{poles of $T^{-1}$}\}\nonumber\\
		&\subseteq Z' \cup\{\text{roots of $\ell_r$}\}\cup Z\nonumber\\
		&\subseteq \bigcup_{i=1}^I\left( \left\{\a_{i,1},\ldots, \beta_i-1\right\}\cup\{\a_{i,1},\a_{i,2},\ldots, \a_{i,J_i}\}\cup\{\a_{i,1}+1,\ldots, \beta_i\}\right)\nonumber\\
		&\subseteq\bigcup_{i=1}^I\,\{\a_{i,1},\a_{i,1}+1,\ldots, \beta_i\};
	\end{align}
	here we use the assumption that $\a_{i,J_i}\leq\beta_i$.
	
	On the other hand, 	by Lemma~\ref{Lem: left bound}, we know that for each $1\leq i \leq I$, $U$ is a local integral bases at every $\a \leq \a_{i,1}$ with $\a \in \a_i + \set Z$. So by construction, $W$ is a local integral basis at every $\alpha\leq \beta_i$ with $\a \in \a_i + \set Z$. By Lemma~\ref{Lem: integral after shift}, $SW$ is a local integral basis at every $\alpha\leq \beta_i-1$. Thus $e_w$ has no roots in $Z'$ because $SW$ and $W$ generate the same $C(x)_\alpha$-module for all $\alpha\in Z'$. Combining this fact and the relation in~\eqref{EQ:roots of e_w}, we get
	\[\{\text{roots of $e_w$} \}\subseteq\{\beta_i \mid 1\leq i\leq I\}.\]
	Since $\beta_i\in\a_i + \set Z$ and $\a_i+\set Z\neq \a_j + \set Z$ for all $1\leq i\neq j\leq I$, it follows that $e_w$ is shift-free.%any two roots of $e_w$ do not have an integer distance. So $e_w$ is shift-free.
	
	Since both $W$ and $SW$ are $C(x)$-vector space bases of $A$, the matrix $M_w$ is invertible over $C(x)$.
	
	\eqref{it:f_w}\, Note that
	\[U =\begin{pmatrix}
		-\frac{\ell_1}{\ell_0} & -\frac{\ell_2}{\ell_0}&\cdots & -\frac{\ell_{r}}{\ell_0}\\
		1    &  &    &  0\\
		&\ddots    & &\vdots  \\
		&   & 1      &0\\
	\end{pmatrix} SU =: \frac{1}{f_u}N_uSU,
	\]
	where $f_u \in C[x]$ and $N_u\in C[x]^{r\times r}$. So every root of $f_u$ must be a root of the trailing coefficient $\ell_0$. By~\eqref{EQ:SW-W}, we have
	\[W=T\left(\frac{1}{e_u}M_u\right)^{-1}\sigma(T)^{-1}SW=T\frac{1}{f_u}N_u\sigma(T)^{-1}SW =  \frac{1}{f_w}N_w SW.\]
	Using the same argument as in~\eqref{it:e_w}, we get
	\begin{align}\label{EQ:roots of f_w}
		\{\text{roots of $f_w$} \}&\subseteq \{\text{poles of $T$}\}\cup\{\text{roots of $f_u$}\}\cup\{\text{poles of $\sigma(T)^{-1}$}\}\nonumber\\
		&\subseteq Z \cup\{\text{roots of $\ell_0$}\}\cup Z'\nonumber\\
		%&\subseteq \cup_{i=1}^I\left( \{\a_{i,1},\ldots, \beta_i-1\}\cup\{\a_{i,1},\a_{i,2},\ldots, \a_{i,J_i}\}\cup \{\a_{i,1}+1,\ldots, \beta_i\}\right)\nonumber\\
		&\subseteq\bigcup_{i=1}^I\,\{\a_{i,1},\a_{i,1}+1,\ldots, \beta_i\}.	
		\end{align}
	Since both $W$ and $SW$ are local integral bases at every $\a \leq \beta_i-1$ with $\alpha\in \alpha_i +\set Z$, the only possible roots of $f_w$ are $\beta_1,\ldots, \beta_I$. So $f_w$ is shift-free.
\end{proof}

\begin{example}\label{Exam: suitable_ord2}
	 Let $L=(x+2)(x+3)S^2-2(x+2)S+1\in \set C[x][S]$. It has two solutions $b_1= \frac{1}{x!}$ and $b_2=\frac{1}{(x+1)!}.$ All roots of $\ell_0\ell_2$ are $-3,-2 \in \set Z$. Let $\beta=-2$ and $Z=\{-2\}$. A
	 suitable basis of $A=\set C(x)[S]/\<L>$ at $\beta$ is $W=(1,(x+2)S)$. It is a local integral basis of $A$ at $Z$ and generates the same $\set C(x)_\alpha$-module as $U=\{1,S\}$ for all $\alpha\in\set C\setminus Z$. Then by Theorem~\ref{Thm: suitable basis},  $e_w$ is shift-free and $\beta$ is the only possible root of $e_w$. Indeed, we have $e_w=x+2$ and $M_w = \binom{\hphantom{-}0\,\hphantom{-}1}{-1\,\hphantom{-}2}$. Moreover, $f_w=1$ is shift-free and $N_w =\binom{2x + 4\hphantom{-}-x - 2}{x + 2\hphantom{-}\hphantom{--}0}$.
\end{example}

\section{Generalized Abramov-Petkov\v sek reduction}\label{sec: AP}
%\section{Reduction of the dispersion}
In this section, we use a suitable basis to decompose an element $f\in A$ into a summable part and a remainder such that the denominator of this remainder is shift-free. The lemma below presents a matrix version of Abramov-Petkov\v sek reduction formula in~\cite[Lemma 9]{AbramovePetkovsek02}.
\begin{lemma}\label{Lem: AP step}
	Let $W$ be a suitable basis of $A$ at $\{\beta_1,\ldots,\beta_I\}$. Let $e,\tilde e\in C[x]$ and $M,\tilde M\in C[x]^{r\times r}$ be such that $SW=\frac{1}{e}MW$ and $\frac{1}{\tilde e}\tilde M=(\frac{1}{e}M)^{-1}$. Let $p_i\in C[x]$ be the minimal polynomial of $\beta_i$ over $C$. Let $q\in C[x]$, $a\in C[x]^r$ and $\ell\in\set Z$.
	\begin{enumerate}[(i)]
		\item\label{it:AP-red1} If $\gcd(q, \sigma^j(p_i))=1$ for all $i\in\{1,\ldots, I\}$ and $j\in \set Z$, there exist $g\in A$ and $c\in C[x]^r$ such that
		\begin{equation}\label{EQ: AP-step}
			\frac{aW}{\sigma^\ell(q)}=\Delta g +  \frac{cW}{qe\tilde e}.
		\end{equation}
		\item\label{it:AP-red2} If $q=p_i^m$ with $i\in\{1,\ldots, I\}$ and $m>0$, there exist $g\in A$ and $c\in C[x]^r$ such that~\eqref{EQ: AP-step} also holds.
    \end{enumerate}
\end{lemma}
\begin{proof}
	If $\ell =0$, we only need to take $g=0$ and $c=e\tilde ea$.
	
	If $\ell \geq 1$, we have
	\begin{align}\nonumber
		\frac{aW}{\sigma^\ell(q)} &= \frac{a}{\sigma^\ell(q)}\frac{\tilde M}{\tilde e}\frac{M}{e} W- \frac{\sigma^{-1}(a)}{\sigma^{\ell-1}(q)}\frac{\sigma^{-1}(\tilde M)}{\sigma^{-1}(\tilde e)}W + \frac{\sigma^{-1}(a)}{\sigma^{\ell-1}(q)}\frac{\sigma^{-1}(\tilde M)}{\sigma^{-1}(\tilde e)}W\\\label{EQ:AP step 1}
		&=\Delta\left(\frac{\sigma^{-1}(a)}{\sigma^{\ell-1}(q)}\frac{\sigma^{-1}(\tilde M)}{\sigma^{-1}(\tilde e)}W\right) + \frac{\sigma^{-1}(a)}{\sigma^{\ell-1}(q)}\frac{\sigma^{-1}(\tilde M)}{\sigma^{-1}(\tilde e)}W.
	\end{align}
	By the assumption and Theorem~\ref{Thm: suitable basis}.\eqref{it:f_w}, every irreducible factor of $\tilde e$ is a factor of $p_i$ for some $i\in\{1,\ldots, I\}$. So in either~\eqref{it:AP-red1} or~\eqref{it:AP-red2}, we have $\gcd(\sigma^{\ell-1}(q), \sigma^{-1} (\tilde e)) =1$. By the partial fraction decomposition of rational functions, there exist $v_1, v_2\in C[x]^r$ such that
	\begin{equation}\label{EQ: AP red pfd}
		\frac{\sigma^{-1}(a)}{\sigma^{\ell-1}(q)}\frac{\sigma^{-1}(\tilde M)}{\sigma^{-1}(\tilde e)}W=\frac{v_1W}{\sigma^{\ell-1}(q)} + \frac{v_2W}{\sigma^{-1}(\tilde e)}.
	\end{equation}
	Since
	\[\frac{v_2W}{\sigma^{-1}(\tilde e)} =\Delta\left(-\frac{v_2W}{\sigma^{-1}(\tilde e)}\right) + \frac{\sigma(v_2)MW}{\tilde e e},\]
	combining with \eqref{EQ:AP step 1} and~\eqref{EQ: AP red pfd}, we have
	\[\frac{aW}{\sigma^\ell(q)} \equiv \frac{v_1W}{\sigma^{\ell-1}(q)} + \frac{\sigma(v_2)MW}{\tilde e e} \mod \Delta(A).\]
	By the induction hypothesis, the first summand is congruent to $\frac{c_1W}{qe\tilde e}$ for some $c_1\in C[x]^r$. Setting $c = c_1 + q\sigma(v_2)M$ establishes Equation~\eqref{EQ: AP-step} in the case $\ell\geq 1$.
	
    If $\ell \leq -1$, we have
	\[\frac{aW}{\sigma^{\ell}(q)} = \Delta\left(-\frac{aW}{\sigma^{\ell}(q)}\right) + \frac{\sigma(a)MW}{\sigma^{\ell +1}(q)e}.\]
	When $\ell = - 1$, we take $g=-\frac{aW}{\sigma^{-1}(q)}$ and $c=\tilde e\sigma(a)M$. When $\ell < -1$, note that $\gcd(\sigma^{\ell+1}(q), e) = 1$ by the assumption and Theorem~\ref{Thm: suitable basis}.\eqref{it:e_w}. The induction can be completed in a similar way as in the case $\ell\geq 1$.
\end{proof}
Recall~\cite{AbramovPetkovsek01,AbramovePetkovsek02} that two irreducible polynomials $p,q\in C[x]$ are said to be {\em shift-equivalent} if $p\mid\sigma^k(q)$ for some $k\in \set Z$, denoted by $p\sim_x q$. The relation $\sim_x$ is an equivalence relation. By grouping together the irreducible factors in the same shift-equivalence classes, every polynomial $u\in C[x]$ has a factorization of the form
\[u = \prod_{i=1}^m\prod_{\ell = \ell_1}^{\ell_2} \sigma^\ell(q_i)^{n_{i,\ell}}\]
where $m,n_{i,\ell}\in\set N$, $\ell_1, \ell_2\in \set Z$ with $\ell_1\leq \ell_2$, $q_i\in C[x]$ is irreducible, and $q_i$, $q_{i'}$ are pairwise shift-inequivalent for all $1\leq i\neq i'\leq m$. From such a factorization, applying the partial fraction decomposition of rational functions to all coefficients of $f = \frac{aW}{u}\in A$ with $a\in C[x]^r$, we can decompose $f$ into the form
\begin{equation}\label{EQ: IPFD}
	f = \left(f_0 + \sum_{i=1}^m\sum_{\ell=\ell_1}^{\ell_2}\frac{a_{i,\ell}}{\sigma^\ell(q_i)^{n_{i,\ell}}}\right)W
\end{equation}
where $f_0,a_{i,j,\ell}\in C[x]^r$ and $\deg_x(a_{i,\ell})<\deg_x(q_i^{n_{i,\ell}})$. The {\em degree} of a vector in $C[x]^r$ is the maximal degree of all its entries.
\begin{prop}\label{Prop: AP-reduction}
	Let $W$ be a suitable basis of $A$ at $\{\beta_1,\ldots, \beta_I\}$. Then any element $f\in A$ can be decomposed as
	\begin{equation}\label{EQ: AP-reduction}
		f = \Delta g+ h\quad\text{and}\quad h=\frac{cW}{\tilde u},
	\end{equation}
	where $g\in A$, $c\in C[x]^r$, $\tilde u\in C[x]$ and the product $\tilde u\prod_{i=1}^I(x-\beta_i)$ is shift-free.
\end{prop}
\begin{proof}
	 We write $f$ as in~\eqref{EQ: IPFD}. Let $p_i\in C[x]$ be the minimal polynomial of $\beta_{i}$. If $q_i$ is shift-equivalent to $p_{i'}$ for some $i'\in\{1,\ldots, I\}$, we may choose $q_i = p_{i'}$ as a representative in this shift equivalence class. By Lemma~\ref{Lem: AP step}, there exist $g\in A$ and $c_{i,\ell}\in C[x]^r$ such that
	\[f=\Delta g + \left(f_0 + \sum_{i=1}^m\sum_{\ell =\ell_1}^{\ell_2}\frac{c_{i,\ell}}{q_i^{n_{i,\ell}}e\tilde e}\right)W,\]
	where $SW = \frac{1}{e}MW$ and $\frac{1}{\tilde e}\tilde M = (\frac{1}{e}M)^{-1}$. By Theorem~\ref{Thm: suitable basis}, $e\tilde e\prod_{i=1}^I(x-\beta_i)$ is shift-free. Then we get~\eqref{EQ: AP-reduction} by setting $\tilde u =  e\tilde e\prod_{i=1}^mq_i^{n_i}$ with $n_i = \max_{\ell=\ell_1}^{\ell_2}\{n_{i,\ell}\}$ and $c = f_0v + \sum_{i=1}^m\sum_{\ell=\ell_1}^{\ell_2}( vc_{i,\ell})/(q_i^{n_{i,\ell}}e\tilde e)$. By the assumption on suitable bases, the orbits $\beta_i + \set Z$ with $1\leq i \leq I$ are distinct and all conjugates of $\beta_i$ belong to $\{\beta_1,\ldots, \beta_I\}$. So the product $p_{i'}\prod_{i=1}^I (x-\beta_I)$ is shift-free for all $1\leq i'\leq I$. Thus the polynomial $\tilde u$ satisfies the required condition.
\end{proof}

\begin{theorem}\label{Thm: AP-remainder}
	Let $h\in A$ be in~\eqref{EQ: AP-reduction} with $c=(c_1,\ldots,c_r)$ and $\gcd(\tilde u, c_1,\ldots,c_r)=1$. Let $e\in C[x]$ and $M\in C[x]^{r\times r}$ be such that $SW = \frac{1}{e}MW$. If $h$ is summable in $A$, then  $\tilde u$ divides $e$ and $h=\Delta(bW)$ with $b\in C[x]^r$.
\end{theorem}
\begin{proof}
	Suppose $h$ is summable in $A$. There exist $b=(b_1,\ldots,b_r)\in C[x]^r$ and $v\in C[x]$ such that $h = \Delta(\frac{1}{v}bW)$ and $\gcd(v,b_1,\ldots, b_r)=1$. We show that $v$ is a constant. Otherwise there exists $\alpha\in \bar C$ such that $\disp_\alpha(v) \geq 0$. Since $W$ is a suitable basis at $\{\beta_1,\ldots,\beta_I\}$, by Proposition~\ref{Prop: disp after diff} and Theorem~\ref{Thm: suitable basis}, we have that either $\disp_\alpha(\tilde u) \geq \disp_\alpha(v)+1\geq 1$ or $\disp_\alpha((x-\beta_i)\tilde u)\geq 1$ for some $i\in\{1,\ldots, I\}$. This contradicts the fact that $\tilde u\prod_{i=1}^I (x-\beta_i)$ is shift-free. Thus $v$ is a constant in $C$. Then
	\[h=\sigma(v^{-1}b)SW-v^{-1}bW = \left(\frac{1}{e}\sigma(v^{-1}b)M-v^{-1}b\right)W = \frac{cW}{\tilde u}.\]
	So $\tilde u$ divides $e$.
\end{proof}

\begin{corollary}\label{Cor: AP-remainder}
	Let $W$ be a suitable basis of $A$ at $\{\beta_1,\ldots, \beta_I\}$. Let $e\in C[x]$ and $M\in C[x]^{r\times r}$ be such that $SW = \frac{1}{e}MW$. Then any element $f\in A$ can be decomposed as
	\begin{equation}\label{EQ: AP-remainder}
		f = \Delta g+ \frac{1}{d}PW+\frac{1}{e}RW,
	\end{equation}
	where $g\in A$, $d\in C[x]$, $P,R\in C[x]^r$ with $\deg_x(P)<\deg_x(d)$ and $d\prod_{i=1}^I(x-\beta_i)$ being shift-free. Moreover, if $f$ is summable in $A$, then $P=0$.
\end{corollary}
\begin{proof}
	Let $g\in A$ and $h=\frac{cW}{\tilde u}$ be as in~\eqref{EQ: AP-reduction}. Let $\{p_1,\ldots, p_m\}$ be the set of minimal polynomials of the $\beta_i$'s over $C$ for all $1\leq i\leq I$ and let $\tilde u = \tilde u\prod_{i=1}^mp_i^{n_i}$ be such that $\tilde u\in C[x]$ and $\gcd(\tilde u,\prod_{i=1}^mp_i)=1$. Applying the partial fraction decomposition of rational functions to all entries of $c$, we decompose $h$ as
	\begin{equation}\label{EQ:h_decomp_step1}
		h = \left(h_0 + \frac{\tilde c}{\tilde u} +\sum_{i=1}^m\frac{c_i}{p_i^{n_i}} \right)W,
	\end{equation}
	where $h_0, \tilde c, c_i\in C[x]^r$ satisfy $\deg_x(\tilde c)<\deg_x(\tilde u)$ and $\deg_x(c_i)<\deg_x(p_i^{n_i})$. Let $k_i$ be the multiplicity of $p_i$ in $e$. If $n_i>k_i$, by division with remainder, we write $c_i = a_ip_i^{n_i-k_i} + b_i$ such that $a_i,b_i\in C[x]^r$ and $\deg_x(b_i)<\deg_x(p_i^{n_i-k_i})$. Then we have
	\begin{equation}\label{EQ:h_decomp_step2}
		\frac{c_i}{p_i^{n_i}}W= \left(\frac{b_i}{p_i^{n_i}} + \frac{a_i}{p_i^{k_i}} \right)W.
	\end{equation}
	If $n_i\leq k_i$, then~\eqref{EQ:h_decomp_step2} holds for $b_i=0$ and $a_i = c_ip_i^{k_i-n_i}$. Let $d$ be the product of $\tilde u$ and $p_i^{n_i}$ for all $1\leq i\leq m$ with $b_i\neq 0$. Combining \eqref{EQ:h_decomp_step1} and~\eqref{EQ:h_decomp_step2}, we decompose $h$ into two parts:
	\begin{equation*}%\label{EQ:h_decomp}
		h = \left(\frac{\tilde c}{\tilde u} +\sum_{i=1}^m \frac{b_i}{p_i^{n_i}} \right) W +\left(h_0+ \sum_{i=1}^m \frac{a_i}{p_i^{k_i}} \right)W\triangleq  \frac{1}{d}PW + \frac{1}{e}RW,
	\end{equation*}
	where $P=(d/\tilde u)\tilde c  + \sum_{i=1}^m (d/p_i^{n_i})b_i$ and $R=eh_0+\sum_{i=1}^m (e/p_i^{k_i})a_i\in C[x]^r$. Then the product $d\prod_{i=1}^I(x-\beta_i)$ is shift-free and $\deg_x(P)<\deg_x(d)$. Since $f=\Delta g +h$, $f$ is summable in $A$ if and only if $h$ is summable in $A$. By Theorem~\ref{Thm: suitable basis}, $\{p_1,\ldots,p_m\}$ are all possible irreducible factors of $e$. Note that for each $i\in\{1,\ldots,m\}$, the multiplicity of $p_i$ in $d$ is either greater than $k_i$ or equal to $0$. By Theorem~\ref{Thm: AP-remainder}, if $h$ is summable in $A$, then $d$ belongs to~$C$ and hence $P=0$ because $\deg_x(P)<\deg_x(d)$. %Combining Proposition~\ref{Prop: AP-reduction} and Equation~\eqref{EQ:h_decomp}, we get the following corollary.
\end{proof}
\begin{example}
	Let $L\in \set C[x][S]$ be the same operator as in Example~\ref{Exam: suitable_ord2} and $W=(\omega_1,\omega_2) =(1, (x+2)S)$ be a suitable basis of $A=\set C(x)[S]/\<L>$ at $\{-2\}$. Then $e=x+2$ and $M= \binom{\hphantom{-}0\,\hphantom{-}1}{-1\,\hphantom{-}2}$. For \[f=\frac{1}{(x+1)(x+2)}\omega_1+\frac{x}{(x+1)(x+2)}\omega_2,\]
	applying the reduction described above, we can reduce the dispersion of the denominator of $f$ and decompose it into the form
	\[
	f = \Delta\Bigl({\underbrace{\frac{(-1,1)}{x+1}W}_g}\Bigr) +
	\underbrace{\frac1{(x+2)^2}}_{ 1/d}\underbrace{(1,-1)}_PW +
	\underbrace{\frac1{(x+2)}}_{1/e}\underbrace{(-1,2)}_RW.
	\]	
	By Theorem~\ref{Thm: AP-remainder}, $f$ is not summable in $A$ because $d\notin C$ and $P\neq 0$.
\end{example}
In the next section, we shall use a local integral basis at infinity to reduce the degree of $S$ in~\eqref{EQ: AP-remainder}.

\section{Reduction at infinity}\label{sec: reduce degree}
This section can be viewed as a discrete analog of Hermite reduction at infinity for D-finite functions~\cite{chen23}. It is also a generalization of the polynomial reduction in the hypergeometric case~\cite{chen15a}. Let $W=(\omega_1,\dots,\omega_r)$
be a local integral basis at infinity of~$A$. Let $\lambda\in \set Z$ be the minimal integer such that
\[\Delta W=x^{\lambda}MW\quad \text{with}\quad M=(m_{i,j})_{i,j=1}^r\in C(x)_\infty^{r\times r}.\]
For convenience, we define the {\em degree} of a rational function $p/q\in C(x)$ as $\deg_x(p)-\deg_x(q)$, denoted by $\deg_x(p/q)$. %Then $\deg_x(p/q)=-\nu_\infty(p/q)$.
The {\em degree} of a matrix in $C(x)^{r\times r}$ is the maximal degree of all its entries. With this convention, we have $\lambda = \deg_x(x^\lambda M)$.
Let $f=x^k\sum_{i=1}^r a_i\omega_i\in A$ with $k\geq 0$ and $a_1,\dots,a_r\in C(x)_\infty$.
In order to reduce the degree $k$ of the coefficient vector of $f$, we seek
$b_1,\dots,b_r,c_1,\dots,c_r\in C(x)_\infty$ such that
\begin{equation}\label{EQ:hr-infty-goal}
	x^k\sum_{i=1}^ra_i\omega_i =\Delta\left(x^{k+1}\sum_{i=1}^rb_i\omega_i\right)+ x^{k-1}\sum_{i=1}^rc_i\omega_i.
\end{equation}
Using the relation $\Delta (uv)= \Delta(u)v + \sigma(u) \Delta(v)$ for all $u\in C(x)$ and $v\in A$, we get

 \begin{equation*}
 	x^k\sum_{i=1}^r a_i \omega_i
 	=\sum_{i=1}^r \left(\Delta(b_i) x^{k+1}\omega_i + \sigma(b_i)\Delta\left(x^{k+1}\omega_i \right) + c_ix^{k-1} \omega_i\right)
 \end{equation*}
Multiplying by $x^{-k}$ and rewriting the difference of $x^{k+1}\omega_i$, we obtain
\begin{align}\label{EQ:hr-expand1}
	&\sum_{i=1}^r a_i \omega_i
	=\sum_{i=1}^r \left(\Delta(b_i) x\omega_i + \sigma(b_i)x^{-k}\Delta\left(x^{k+1}\omega_i \right) + c_ix^{-1} \omega_i\right)\\\label{EQ:hr-expand2}
	&~~=\sum_{i=1}^r \left(\Delta(b_i) x\omega_i + \sigma(b_i)x^{-k}\Delta(x^{k+1}) \omega_i +x^{-k}\sigma(b_ix^{k+1})x^{\lambda}\sum_{j=1}^rm_{i,j} \omega_j+ c_ix^{-1} \omega_i\right).
\end{align}
Since $b_i\in C(x)_\infty$, we have $\Delta(b_i)x \in x^{-1} C(x)_\infty$. For example, $\Delta(1+\frac{1}{x} + \cdots)x=-\frac{1}{x}+\cdots$. Also $x^{-k}\Delta(x^{k+1}) = (k+1) + \binom{k+1}{2} x^{-1} +\cdots \equiv  (k+1)\mod x^{-1}$. So if $\lambda <-1$, then Equation~\eqref{EQ:hr-expand2} can be reduced modulo $x^{-1}C(x)_\infty$:
\begin{equation}\label{EQ:hr-simple-case}
	\sum_{i=1}^r a_i \omega_i %&\equiv \sum_{i=1}^n\sigma(b_i)x^{-k}\Delta(x^{k+1})\omega_i\mok x^{-1}\\
	\equiv \sum_{i=1}^r(k+1)\sigma(b_i)\omega_i\mod x^{-1}.
\end{equation}
It follows that $b_i \equiv \sigma^{-1}((k+1)^{-1} a_i)= (k+1)^{-1}\sigma^{-1}(a_i)  \equiv (k+1)^{-1}a_i\mod x^{-1}$ is the unique solution of~\eqref{EQ:hr-simple-case} in $C(x)_\infty/\<x^{-1}>$.
If $\lambda\geq -1$, then multiplying~\eqref{EQ:hr-expand1} by $x^{-\lambda-1}$ and reducing this equation modulo $x^{-\lambda-2}$ yields
\begin{equation}\label{EQ:hr-mod}
	\sum_{i=1}^r x^{-\lambda-1} a_i\omega_i \equiv \sum_{i=1}^r \sigma(b_i)x^{-k-\lambda-1}\Delta\left(x^{k+1}\omega_i\right)\mod x^{-\lambda-2}.
\end{equation}
Let $\psi_i := x^{-k-\lambda-1}\Delta\left(x^{k+1}\omega_i\right)$ for $i=1,\ldots, r$. To reduce the degree $k$, we have to show that Equation~\eqref{EQ:hr-mod} always has a solution $(b_1,\ldots, b_r)$ in $\left( C(x)_\infty/\<x^{-\lambda-2}>\right)^r$. The solution $b=(b_1,\ldots, b_r)$ in $\left(C(x)_\infty/\<x^{-\lambda-2}>\right)^r$ is of the form $b_i = b_{i,0} + b_{i,1}\frac{1}{x} + \cdots + b_{i,\lambda+1}\frac{1}{x^{\lambda+1}}$ with $b_{i,j}\in C$. In practice, we hope that $x^{k+1} \sum_{i=1}^rb_i\omega_i$ in~\eqref{EQ:hr-infty-goal} has only polynomial coefficients and therefore assume $k\geq\max\{0,\lambda\}$.

\begin{prop}\label{Prop:psi}
	Let $W=\{\omega_1, \ldots, \omega_r\}$ be a local integral basis at infinity of $A$. Let $\lambda\in \set Z$ be the minimal integer such that $\Delta(W) = x^\lambda MW$ with $M\in  C(x)_\infty^{r\times r}$. For an integer $k\geq 0$, we define $\psi_i := x^{-k-\lambda-1}\Delta(x^{k+1}\omega_i)$. If $\lambda \geq -1$, then
	\[\sum_{i=1}^r C(x)_\infty \psi_i \subseteq \OO_\infty \subseteq x^{\lambda+1}\sum_{i=1}^r C(x)_\infty \psi_i.\]
	In particular, when $\lambda=-1$, we have $\sum_{i=1}^r C(x)_\infty\psi_i =\OO_\infty$. In this case, $\{\psi_1,\ldots, \psi_r\}$ forms a local integral basis at infinity.
\end{prop}
\begin{proof}
	We prove this proposition using the same technique as in \cite[Lemma 10]{chen17a} and~\cite[Proposition 6]{chen23}. To show $\sum_{i=1}^r C(x)_\infty \psi_i \subseteq \OO_\infty$, we only need to show that for every $i=1,\ldots,r$, the element $\psi_i$ is integral at infinity. Expanding $\psi_i$, we get $\psi_i = x^{-k-\lambda-1}\Delta(x^{k+1}) \omega_i +x^{-k-\lambda-1}\sigma(x^{k+1}) \Delta(\omega_i)$. Since $\Delta(W) =x^{\lambda} MW$ and $M\in C(x)_\infty^{r\times r}$, it follows that the second term in $\psi_i$ is integral at infinity. The first term in $\psi_i$ is also integral at infinity because  $\lambda+ 1\geq 0$. So $\psi_i$ is integral at infinity.
	
	Next we shall prove $\OO_\infty \subseteq x^{\lambda+1}\sum_{i=1}^r  C(x)_\infty \psi_i$. Suppose to the contrary that there exists an element $f\in \OO_\infty\setminus x^{\lambda+1}\sum_{i=1}^r C(x)_\infty\psi_i$. Furthermore, we can find such an element $f$ of the form
	\[f= x^{\lambda+2}\sum_{i=1}^r \sigma(c_i) \psi_i\quad\text{with } c_i\in  C(x)_\infty \text{ and } \nu_\infty(c_i)=0 \text{ for some } i.\]
	Taking $f=0$, we shall prove that the $\{\psi_1,\ldots,\psi_r\}$ is linearly independent.

        Further let $g=x^2\sum_{i=1}^r\Delta(c_i)\omega_i$, which is integral at infinity for the same reason as between \eqref{EQ:hr-expand2} and~\eqref{EQ:hr-simple-case}. Then also their sum
	\begin{align*}
		f+g &= x^{-k+1}\sum_{i=1}^r \left(\sigma(c_i)\Delta\bigl(x^{k+1}\omega_i\bigr)
		+ \Delta(c_i)x^{k+1}\omega_i \right) \\
		&= x^{-k+1} \sum_{i=1}^r \Delta\bigl(c_ix^{k+1}\omega_i\bigr)=x^{-k+1}\Delta\left(x^{k+1}h\right)
	\end{align*}
	must be integral, where $h=\sum_{i=1}^rc_i\omega_i$.
	
		Since $\{\omega_1,\ldots,\omega_r\}$ is a local integral basis at infinity, by \cite[Lemma 2]{chen23} (it also works in the setting of P-recursive sequences), we have $0 \leq \val_\infty(h) <1$.  Note that
	\[\val_\infty(x^{k+1}h)= -k-1+\val_\infty(h)\leq -1 + \val_\infty(h)< 0;\]
	here we use the assumption that $k\geq 0$, because $k=-1$ and $\val_\infty(h)=0$ imply that $\val_\infty(x^{k+1}h)=0$. Since $\val_\infty(x^{k+1}h)\neq0$, by Proposition~\ref{Prop:val_under_diff} we get
	\[\val_\infty(x^{-k+1} \Delta(x^{k+1}h))\leq k-1-k-1+ \val_\infty(h)+1=\val_\infty(h)-1<0.\]
	So $x^{-k+1}\Delta(x^{k+1}h)=f+g$ is not locally integral at infinity, which contradicts the integrality of~$f$. Hence $\OO_\infty \subseteq x^{\lambda+1}\sum_{i=1}^r C(x)_\infty \psi_i$.
\end{proof}
Note that the shift operator $\sigma$ does not change the valuation $\nu_\infty$ of a rational function. The following theorem follows with the same proof as in~\cite[Theorem 7]{chen23}.
\begin{thm}\label{Thm:hr}
	Using the same notations as in Proposition~\ref{Prop:psi}, let $k\geq0$ and $\lambda \geq -1$. Then for any $a_1,\ldots, a_r\in C(x)_\infty$, the linear system%~\eqref{EQ:hr-v}
	\begin{equation*}%\label{EQ:linear-sys}
		\sum_{i=1}^r x^{-\lambda-1}a_i \omega_i = \sum_{i=1}^r \sigma(b_i)\psi_i
	\end{equation*}
	has a solution $(b_1, \ldots, b_r)$ in $\left(C(x)_\infty/\<x^{-\lambda-2}>\right)^r$.
\end{thm}
\if t\completeproof
\blue
\begin{proof}
	By Proposition~\ref{Prop:psi}, the $C(x)_\infty$-module generated by \[\left\{x^{-\lambda-1}\psi_1, \ldots, x^{-\lambda-1}\psi_r\right\}\]
	contains a submodule $\OO_\infty$ of rank $r$. So $\{\psi_1, \ldots, \psi_r\}$ is linearly independent over $C(x)$. Then there exist $t_1,\ldots, t_r\in  C(x)$ such that
	$\sum_{i=1}^r x^{-\lambda-1} a_i \omega_i =\sum_{i=1}^r \sigma(t_i)\psi_i$.
	
	To find a solution $b_i$, we have to show that $t_i\in C(x)_\infty$ for all $i=1,\ldots, r$.  Since $a_i\in C(x)_\infty$ and the $\omega_i$'s are integral at infinity, the element
	\[\sum_{i=1}^ra_i\omega_i =x^{\lambda+1}\sum_{i=1}^r\sigma(t_i)\psi_i\]
	is integral at infinity. By Proposition~\ref{Prop:psi},
	\[x^{\lambda+1}\sum_{i=1}^r\sigma(t_i)\psi_i \in \OO_\infty\subseteq x^{\lambda+1}\sum_{i=1}^r C(x)_\infty\psi_i.\]
	Then $\sum_{i=1}^r\sigma(t_i)\psi_i \in \sum_{i=1}^r C(x)_\infty\psi_i.$ Since $\{\psi_1, \ldots, \psi_r\}$ are linearly independent over $C(x)$, we have $\sigma(t_i)\in C(x)_\infty$ for all $i$. Then $t_i\in C(x)_\infty$ for all $i$ as claimed because $\sigma$ does not change the valuation of a rational function in $x^{-\lambda-2}C(x)_\infty$. Note that \[\text{$b_i\equiv t_i\mod x^{-\lambda-2}$\quad if and only if\quad $\sigma(b_i)\equiv \sigma(t_i)\mod x^{-\lambda-2}$}.\]
	So $b_i\equiv t_i\mod x^{-\lambda-2}$ is a solution in $x^{-\lambda-2}C(x)_\infty$.
\end{proof}
\black
\fi
\begin{example}
		Let $L\in \set C[x][S]$ be the same operator as in Example~\ref{Exam: suitable_ord2}. A local integral basis at infinity of $A=\set C(x)[S]/\<L>$ is given by $\omega_1=(x^4+2x^3)S-x^3-x^2$ and $\omega_2=(x^3+2x^2)S-x^2$. Then
		\[\begin{pmatrix}
			\Delta \omega_1\\
			\Delta \omega_2
			\end{pmatrix} = x^\lambda\begin{pmatrix}
			\frac{-x^3-x^2+2x+1}{x^2(x+2)}&0\\
			0&\frac{-x^3-x^2+2x+1}{x^2(x+2)}
		\end{pmatrix}
		\begin{pmatrix}
			\omega_1\\
			\omega_2
		\end{pmatrix}\]
		with  $\lambda =0 $. We apply the reduction at infinity to  \[f=\frac{x^3+2}{x^2(x+2)}\omega_1+\frac{2x^3+3}{x^2(x+2)}\omega_2.\]
		So we start with $k = 0$, $a_1=(x^3+2)/(x^2(x+2))$ and $a^2 = (2x^3+3)/(x^2(x+2))$.
		Then Equation~\eqref{EQ:hr-mod} leads to the following linear system for the unknowns $b_1$ and $b_2$:
		\[(x^{-1}a_1,x^{-1}a_2)\equiv (\sigma(b_1),\sigma(b_2))\begin{pmatrix}
			\frac{-x^4 - x^3 + 3x^2 + 3x + 1}{x^3(x+2)} &0\\
			0& 	\frac{-x^4 - x^3 + 3x^2 + 3x + 1}{x^3(x+2)}
		\end{pmatrix}\mod x^{-2}.\]
		By Theorem~\ref{Thm:hr}, this system always has a solution. Indeed, we find a solution $b_1=-\frac{1}{x}$ and $b_2 = -\frac{2}{x}$. Then one step of the reduction at infinity simplifies $f$ to
		\[f = \Delta\Bigl(-\omega_1-2\omega_2\Bigr) +\frac{x^2-2x+1}{x^2(x+2)}\omega_1 + \frac{x^2-2x+2}{x^2(x+2)}\omega_2.\]	
\end{example}

\begin{remark}
	%what does “Fuchsian” mean?
	Let $W$ be a local integral basis at infinity of $A=C(x)[S]/\<L>$. Let $\lambda \in \set Z$ be the minimal integer such that $\Delta W = x^\lambda MW$ with $M\in C(x)_\infty^{r\times r}$. If the operator $L$ (of order $r$) admits $r$ linearly independent solutions in
	\[\bar C[[[x^{-1}]]] := \bigcup _{\alpha\in \bar C}x^\alpha \bar C[[x^{-1/v}]][\log(x)],\]
    then $\lambda \leq -1$. This is similar to the Fuchsian D-finite case, see~\cite[Lemma 4]{chen17a}.
\end{remark}

Let $W=\{\omega_1, \ldots,\omega_r\}$ be a local integral basis at infinity. Let $\lambda\in \set Z$ be the minimal integer such that $\Delta W = x^\lambda MW$ with $M\in  C(x)_\infty^{r\times r}$. By a repeated application of the reduction at infinity, one can reduce the degree in $x$ as far as possible and decompose $f\in A$ as
\begin{equation}\label{EQ:hr-remainder-infty}
	f= \Delta g + h\quad \text{with}\quad h =\sum_{i=1}^r h_i\omega_i
\end{equation}
where $g\in A$, $h_i\in C(x)$ with $\deg_x(h_i) < \max\{0,\lambda\}$ for all $i$ and  the coefficients of $g$ are polynomials. The following lemma follows from Proposition~\ref{Prop:val_under_diff} with the same proof as in~\cite[Lemma 10]{chen23}.%gives an upper bound for the degree of any hypothetical integral of $h$ in $A$.
\begin{lemma}\label{LEM:degree}
	Let $h \in A$ be as in~\eqref{EQ:hr-remainder-infty}. If $h$ is summable in $A$, then $h=\Delta(\sum_{i=1}^r b_i\omega_i)$ with $b_i\in C(x)$ and $\deg_x(b_i)\leq\max\{0,\lambda\}$ for all $i\in\{1,\ldots,r\}$.
\end{lemma}
\if t\completeproof\blue
\begin{proof}
	Suppose $h$ is summable in $A$. Then there exists $H=\sum_{i=1}^r b_i\omega_i\in A$ with $b_i\in C(x)$ such that $h=\Delta H$. By~\eqref{EQ:hr-remainder-infty},  we know the coefficients of $h$ satisfy $\deg_x(h_i) < \max\{0,\lambda\}$, which implies $\nu_\infty(h_i)> \min\{0,-\lambda\}$. Since the $\omega_i$'s are integral at infinity, it follows that $\val_\infty(h)\geq \min_{i=1}^r\{v_\infty(h_i) + \val_\infty(\omega_i)\}> \min\{0,-\lambda\}$.
	
	We want to show that $\deg_x(b_i)\leq \max\{0,\lambda\}$ for all $i$, i.e., $\nu_\infty(b_i) \geq \min\{0,-\lambda\}$ for all~$i$. Suppose to the contrary that $\tau:=\min_{i=1}^r\{\nu_\infty(b_i)\} < \min\{0,-\lambda\}$. Since $\{\omega_1,\ldots, \omega_r\}$ is a local integral basis at infinity, by~\cite[Lemma 2]{chen23} we get $\val_\infty(H) < \tau +1 \leq \min\{0,-\lambda\}$. So $\val_\infty(H)\neq 0$. By Proposition~\ref{Prop:val_under_diff} we have
	\[\val_\infty(h) \leq \val_\infty(H) + 1\leq \min\{0,-\lambda\}.\]
	This leads to a contradiction. So $\deg_x(b_i)\leq \max\{0,\lambda\}$ for all $i$.
\end{proof}
\fi\black
\section{An additive decomposition}\label{sec:add}
Now we combine the generalized Abramov-Petkov\v sek reduction and the reduction at infinity to decompose a P-recursive sequence $f$ as $f=\Delta g + h$ such that $f$ is summable if and only if the remainder $h$ is zero. This procedure is similar to the hypergeometric case~\cite{chen15a}, P-recursive case~\cite{vanderHoeven18b,BrochetBruno23} and the differential case~\cite{bostan13a,chen16a,chen17a, vanderHoeven21,bostan18a,chen23}. It provides an alternative method for solving the accurate integration problem for P-recursive sequences~\cite{AbramovVanHoeij1999}.

%Since there may not exist a basis of $A=C(x)[S]/\<L>$ that is a local integral basis at all $\a\in  \bar C\cup \{\infty\}$,
In this section, we need two bases to perform the generalized Abramov-Petkov\v sek reduction and the reduction at infinity, respectively. Let $W=(\omega_1,\ldots,\omega_r)$ be a suitable basis of~$A$ at $\{\beta_1,\ldots, \beta_I\}$ that is normal at infinity. By the same proof as in~\cite[Lemma 18]{chen17a}, there exists $T = \diag\bigl(x^{\tau_1}, \ldots, x^{\tau_r}\bigr) \in C(x)^{r\times r}$ with $\tau_i\in \set Z$ such that $V := TW$ is a local integral basis at infinity. Let $e, a\in C[x]$ and $M, B\in C[x]^{r \times r} $ be such that $SW = \frac{1}{e}MW$ and $\Delta V=\frac{1}{a}BV$. Since the difference of $V$ is
$\Delta V = (S-1)TW = (\frac{1}{e}\sigma(T)M -T)T^{-1}V,$
we may assume that $a=x^{\lambda_1}(x+1)^{\lambda_2} e$ for some $\lambda_1,\lambda_2\in \set N$. For $\mu, \delta\in \set Z$ with $\mu \leq \delta$, we define a subspace of Laurent polynomials in $C[x,x^{-1}]$ as $C[x]_{\mu,\delta}:= \{\sum_{i=\mu}^{\delta} a_i x^i\,|\,a_i\in C\}$. The following theorem decomposes any $f\in A$ into a summable part and a remainder such that the remainder belongs to a finite dimensional vector space over~$C$.

\begin{theorem}\label{Thm:add decomp}
	Let $W, V\in A^r$ be as described above. Then any element $f\in A$ can be decomposed into
	\begin{equation}\label{EQ:add}
		f = \Delta g + \frac{1}{d} PW + \frac{1}{a} QV,
	\end{equation}
	where $g\in A$, $d\in C[x]$, $P\in C[x]^r$, $Q\in C[x]_{\mu,\delta}^r$ with $\deg_x(P) < \deg_x(d)$, $\mu = \min\{-\tau_1, \ldots, -\tau_r, 0\}$, $\delta = \max\{\deg_x(a),\deg_x(B)\}-1$ and the product $d\prod_{i=1}^I(x-\beta_i)$ being shift-free. Moreover, $f$ is summable in $A$ if and only if $R=0$ and
	\[\frac{1}{a}QV\in \Delta(U) \quad \text{with} \quad U =\left\{cV \mid c\in C[x]_{\mu, \delta'}^r\right\},\]
	where $\delta'=\max\{0, \deg_x(B) -\deg_x(a)\}$.
\end{theorem}
\begin{proof}
	After performing the generalized Abramov-Petkov\v sek reduction in Section~\ref{sec: AP}, it follows from Corollary~\ref{Cor: AP-remainder} that
		\begin{equation}\label{EQ:add1}
		f = \Delta g+ \frac{1}{d}PW+\frac{1}{e}RW,
	\end{equation}
	where $g\in A$, $d\in C[x]$, $P,R\in C[x]^r$ with $\deg_x(P)<\deg_x(d)$ and $d\prod_{i=1}^I(x-\beta_i)$ being shift-free. We rewrite the last summand in~\eqref{EQ:add1} with respect to the basis~$V$:
	\begin{equation*}%\label{EQ:W-V}
		\frac{1}{e} RW = \frac{1}{a} \tilde{R}V,
	\end{equation*}
	where $\tilde{R} = x^{\lambda_1}(x+1)^{\lambda_2} R T^{-1} \in x^\mu C[x]^r$. Since $V$ is a local integral basis at infinity, using the reduction at infinity in Section~\ref{sec: reduce degree}, we obtain from~\eqref{EQ:hr-remainder-infty} that
	\begin{equation}\label{EQ:S1S2}
		\frac{1}{e} RW = \Delta(R_1 V) + \frac{1}{a} R_2 V,
	\end{equation}
	where $R_1\in C[x]^r$ and $R_2\in x^\mu C[x]^r$ satisfies
	\[\deg_x\left(\tfrac{R_2}{a}\right) \leq \max\left\{0, \deg_x\left(\tfrac{B}{a}\right)\right\}-1.\] This implies that $\deg_x(R_2)\leq \max\{ \deg_x(a),\deg_x(B)\}-1=\delta$. Thus $R_2\in C[x]_{\mu,\delta}^r$ and we finally obtain the decomposition~\eqref{EQ:add} by setting	$g = \tilde{g} + R_1 V$ and $Q = R_2$.	
	
	For the last assertion, assume that $f$ is summable (the other direction of
	the equivalence holds trivially). Then Corollary~\ref{Cor: AP-remainder} implies that $P=0$. Hence the last
	summand in~\eqref{EQ:add1} is also summable. By Theorem~\ref{Thm: AP-remainder}, there exists $b\in C[x]^r$ such that $\frac{1}{e} RW = \Delta(bW)$. By Equation~\eqref{EQ:S1S2}, we get
	\[\frac{1}{a} QV= \frac{1}{e}RW - \Delta(R_1 V) = \Delta\!\left(\left(bT^{-1} - R_1\right)V\right)=\Delta (cV),\]
	where $c = bT^{-1} - R_1\in x^{\mu}C[x]^r$. Since $V$ is a local integral basis at infinity, it follows from Lemma~\ref{LEM:degree} that $\deg_x(c) \leq \max\left\{0, \deg_x\left(\frac{B}{a}\right)\right\}=\delta'$.
	Thus $c\in C[x]_{\mu,\delta'}^r$.
\end{proof}

The remaining step is to reduce all summable P-recursive sequences to zero. Note that in Theorem~\ref{Thm:add decomp}, $U$ is a $C$-vector space of dimension $r(\delta'-\mu'+1)$. Since $\Delta$ is $C$-linear, it follows that $\Delta(U)$ is also a finite-dimensional $C$-vector space. Using Gaussian elimination, $f$ in~\eqref{EQ:add} can be further decomposed as
\begin{equation}\label{EQ:add2}
	f = \Delta \tilde g + \frac{1}{d} PW + \frac{1}{a} Q_2V,
\end{equation}
where $\tilde g= g+ g_1$ with $g_1\in U$ and $Q_2\in C[x]_{\mu,\delta}^r$ such that $f$ is summable in $A$ if and only if $P=Q_2=0$. This decomposition~\eqref{EQ:add2} is called an \emph{additive decomposition} of~$f$.

\begin{prop}
	We use the same notations as in Theorem~\ref{Thm:add decomp}. The vector $Q_2$ in~\eqref{EQ:add2} belongs to a $C$-vector space of dimension at most
	\begin{equation*}%\label{EQ: bound_Nv}
			r(\delta-\mu+1)=\max\{\deg_x(a),\deg_x(B)\}+\max\{\tau,0\},
	\end{equation*}
	where $\tau=\max\{\tau_1,\ldots, \tau_r\}$.
\end{prop}

\section{Creative Telescoping}
Let $K(x)[S_t, S_x]$ with $K=C(t)$ be an Ore algebra, in which $S_t$ and $S_x$ are the shift operators with respect to $t$ and $x$, respectively. Let $J$ be a left ideal of $K(x)[S_t, S_x]$ such that $A=K(x)[S_t, S_x]/J$ is a $K(x)$-vector space of dimension $r$. Assume that there exists a cyclic vector $\gamma$ with respect to $S_x$. This means that $\{\gamma, S_x\gamma, \ldots, S_x^{r-1}\gamma\}$ is a basis of $A$ over $K(x)$. Then $\gamma$ is annihilated by $L$ and $S_t - u_t$ for some $L, u_t\in K(x)[S_x]$. We further assume that $L=\ell_0 + \ell_1S_x+\cdots+\ell_rS_x^r\in K[x][S_x]$ with $\ell_0\ell_r\neq 0$. Every element $f$ in $A$ can be uniquely written as $P_f\gamma+J$ for some $P_f \in K(x)[S_x]$. The map sending $f$ to $P_f+\<L>$ gives an isomorphism from $A$ to $K(x)[S_x]/\<L>$ as a $K(x)[S_x]$-module. Using this isomorphism, for any $f\in A$, we can apply our additive decomposition to test whether $f$ is summable (in $x$). If $f\in A$ is not summable, one can ask to find a nonzero operator $T \in C(t)[S_t]$ (free of $x$) such that $T(f)$ is summable. Such an operator $T$ if it exists is called a \emph{telescoper} for $f$. Applying the additive decomposition with respect to $x$ in Section~\ref{sec:add} to $S_t^i f\in A$ yields that
\begin{equation}
	S_t^if = \Delta_x(g_i)+h_i\quad\text{with}\quad h_i = \frac{1}{d}P_iW + \frac{1}{a
	}Q_iV
\end{equation}
where $g_i\in A$, $d\in K[x]$, $P_i \in K[x]^r$, $Q_i \in K[x]_{\mu,\delta}^r$ with $\deg_x(P_i) < \deg_x(d)$. Moreover, $S_t^if$ is summable in $A$ if and only if $h_i=0$. To compute telescopers for bivariate P-recursive sequences, we are confronted with the same problem as in the hypergeometric case~\cite{chen15a}. The sum of two remainders in the additive decomposition may not be a remainder. Using the same technique as in the hypergeometric case~\cite[Section 5]{chen15a},  one can replace $h_i$ by $\tilde h_i$ such that the linear combination of the $\tilde h_i$'s over $C(t)$ is still a remainder. We shall explain how it works by the following example. If there exist $c_0, c_1,\ldots, c_{m}\in C(t)$ such that $\sum_{i=0}^m c_i \tilde h_i=0$, then $T=\sum_{i=0}^mc_iS_t^i$ is a telescoper for $f$. This approach is the method of reduction-based telescoping and was developed for various classes of functions~\cite{bostan10b,bostan13a,chen15a,chen16a, chen17a, bostan18a, vanderHoeven21,BrochetBruno23,vanderHoeven18b}. %If $\partial_t=D_t$

\begin{example}
	Let $F=x+t^2+\frac{1}{x!}$, which is annihilated by
	\[L =(x+2)(x^2+(t^2+1)x+1)S_x^2 - (x^3+(t^2+5)x^2+(3t^2+7)x+t^2+4)S_x +x^2+(t^2+3)x+t^2+3\]
	and
	\[S_t -\frac{(2t+1)(x+1)}{x^2+(t^2+1)x+1}S_x- \frac{x^2+(t^2+2)x-2t}{x^2+(t^2+1)x+1}.\]
	A suitable basis of $A=K(x)[S_x]/\<L>$ with $K=\set C(t)$ at $\{\beta_1,\beta_2,\beta_3\}$ is
	\[W=(\omega_1,\omega_2)=\left(\frac{(t^2-1)}{x^2+(t^2+1)x+1}((x+t^2)S_x-(x+t^2+1)),\,\frac{1}{x^2+(t^2+1)x+1}(S_x-\frac{x+t^2}{t^2-1})\right),\]
	where $\beta_1=-2$, $\beta_2$ and $\beta_3$ are distinct roots of $x^2+(t^2+1)x+1$. A local integral basis at infinity is $V=(v_1,v_2) = (x\omega_1,\omega_2) $. Then
	\begin{align*}
	&\begin{pmatrix}
		S_x\omega_1\\
		S_x\omega_2
	\end{pmatrix} = \frac{1}{x+2}\begin{pmatrix}
	1&0\\
	-\frac{x+1}{(t^2-1)^2}&x+2
	\end{pmatrix}\begin{pmatrix}
	\omega_1\\
	\omega_2
	\end{pmatrix},
	&\begin{pmatrix}
		S_t\omega_1\\
		S_t\omega_2
	\end{pmatrix} = \begin{pmatrix}
		\frac{t(t+2)}{t^2-1}&0\\
		0&\frac{t^2-1}{t(t+2)}
	\end{pmatrix}\begin{pmatrix}
		\omega_1\\
		\omega_2
	\end{pmatrix}, \\
	&\begin{pmatrix}
		\Delta_xv_1\\
		\Delta_xv_2
	\end{pmatrix} = \frac{1}{x(x+2)}\begin{pmatrix}
		-x^2-x+1&0\\
		-\frac{x+1}{(t^2-1)^2}&0
	\end{pmatrix}\begin{pmatrix}
		v_1\\
		v_2
	\end{pmatrix}, 	
	&\begin{pmatrix}
		S_tv_1\\
		S_tv_2
	\end{pmatrix} = \begin{pmatrix}
		\frac{t(t+2)}{t^2-1}&0\\
		0&\frac{t^2-1}{t(t+2)}
	\end{pmatrix}\begin{pmatrix}
		v_1\\
	    v_2
	\end{pmatrix}.
	\end{align*}
	We compute a minimal telescoper for $\frac{1}{x+t}F$, which corresponds to $f =\frac{1}{x+t}\in A$. Its representation in the bases is
	\begin{align*}
		f&=\frac{1}{x+t}(\omega_1-(t^2-1)(x+t^2)\omega_2)\\
		&=\frac{1}{x+t}(1,-(t-1)^2t(t+1))W+(0,-(t^2-1))V.
	\end{align*}
    After the reduction at infinity, we obtain
    \[f=\Delta_x((\tfrac{1}{t^2-1},-(t^2-1)x)V) +\frac{1}{x+t}(1,-(t-1)^2t(t+1))W+\frac{1}{x(x+2)}(-\tfrac{1}{t^2-1}(x+2),0)V.\]
	Since $a=x(x+2)$, $\mu=-1$, $\delta=1$ and $\delta'=0$, a basis of the space of all summable sequences in the form $\frac{1}{a}QV$ with $Q\in K[x]_{\mu,\delta}^2$ is given by $\Delta_x(\frac{1}{x}v_1)=-\frac{x+1}{x(x+2)}v_1$.
	An additive decomposition of $f$ is
	\[f=\Delta_x(g_0)+h_0\quad\text{with}\quad h_0=\frac{1}{x+t}(1,-(t-1)^2t(t+1))W+\frac{1}{x(x+2)}(-\tfrac{1}{t^2-1},0)V,\]
	where $g_0=(\tfrac{1}{t^2-1},-(t^2-1)x)V+(\frac{1}{(t^2-1)x},0)V$.
	Set $\tilde h_0=h_0$. Next we consider $S_tf$, which has an additive decomposition
	\[S_tf =\Delta_x(S_tg_0)+h_1 \quad\text{with}\quad h_1=\frac{1}{x+t+1}(\tfrac{t(t+2)}{t^2-1},-(t-1)t(t+1)^2)W+\frac{1}{x(x+2)}(-\tfrac{1}{t^2-1},0)V.\]
	The sum of two remainders $h_0$ and $h_1$ is not a remainder, because $(x+t)(x+t+1)$ is not shift-free with respect to $x$. Applying the generalized Abramov-Petkov\v sek reduction to $h_1$ by the formula~\eqref{EQ:AP step 1}, we get
	\[S_tf=\Delta_x(\cdots)+\frac{1}{x+t}(\tfrac{t}{t^2-1}(x+t+2),- (t-1)t(t+1)^2)W+\frac{1}{x(x+2)}(-\tfrac{1}{t^2-1},0)V,\]
	which has another additive decomposition
	\[S_tf=\Delta_x(\tilde g_1)+\tilde h_1\quad\text{with} \quad \tilde h_1=\frac{1}{x+t}(\tfrac{2t}{t^2-1},-(t-1)t(t+1)^2)W+\frac{1}{x(x+2)}(\tfrac{1}{t+1},0)V,\]
	where $\tilde g_1\in A$. Now the sum of $\tilde h_0$ and $\tilde h_1$ is a remainder. Since $\tilde h_0$ and $\tilde h_1$ are linearly independent over $\set C(t)$, we continue with $S_t^2 f$, $S_t^3f$ and compute their refined additive decompositions $S_t^i(f)=\Delta_x(\tilde g_i)+\tilde h_i$ for $i=2,3$, where $\tilde g_2,\tilde g_3\in A$ and
	\begin{align*}
		\tilde h_2&=\frac{1}{x+t}(\tfrac{t^2+2}{t-1},-(t-1)(t+1)^2(t+2))W+\frac{1}{x(x+2)}(-\tfrac{t^2}{t^2-1},0)V,\\
		\tilde h_3&=\frac{1}{x+t}(-\tfrac{(t+2)(t^3-2t-3)}{t^2-1},-(t+2)(t+3)(t^2-1))W+\frac{1}{x(x+2)}(\tfrac{t^2+t-1}{t-1},0)V.
	\end{align*}
	Now we see the remainders $\tilde h_0, \tilde h_1, \tilde h_2, \tilde h_3$ are linearly dependent over $\set C(t)$, which gives rise to a minimal telescoper $(3t^2+3t+2)S_t^3+(3t^3+3t^2-4t-6)S_t^2 -(6t^3+15t^2+13t+2)S_t+3t^3+9t^2+8t$.
\end{example}

\bibliographystyle{plain}
%\balance
\bibliography{integral}

\end{document}